\documentclass[11pt]{article}
\usepackage{xspace}
\usepackage[margin=1in]{geometry}

\usepackage{subfigure}
\usepackage{times}
\usepackage{sidecap}
\usepackage{enumerate}
\usepackage{tcolorbox}

\usepackage{amsmath,amssymb, amsthm, enumerate,boxedminipage}

\theoremstyle{definition}
\newtheorem*{def*}{Definition}
\newtheorem{dfn}{Definition}
\newtheorem{theorem}{Theorem}[section]
\theoremstyle{definition}

\newtheorem{lemma}{Lemma}

\theoremstyle{definition}
\newtheorem*{prf*}{Proof}
  
\theoremstyle{definition}
\newtheorem*{prfthm*}{Proof of Theorem}

\theoremstyle{theorem}

\newtheorem{claim}[theorem]{Claim}

\newcommand{\Act}{\emph{Active}\xspace}

\newcommand{\efbad}{\ensuremath{\epsilon_0}}

\newcommand{\light}{\emph{light}\xspace}

\newcommand{\fbad}{\ensuremath{1/4}}

\newcommand{\maxact}{\ensuremath{\textit{max}}_a\xspace}
\newcommand{\minact}{\ensuremath{\textit{min}}_a\xspace}

\newcommand{\prom}{\emph{Promise Agreement}\xspace}
\newcommand{\promAlg}{\emph{\sc PromiseAgreement}\xspace}

\newcommand{\whp}{w.h.p.\xspace}

\newcommand{\act}{\emph{active}\xspace}

\def \polylog{\operatorname{polylog}}

\def\eps{{\epsilon}}

\newcommand{\rinput}{\ensuremath{\textit{ready-in}}\xspace}
\newcommand{\rout}{\ensuremath{\textit{ready-out}}\xspace}
\newcommand{\vall}{\ensuremath{\textit{value}}\xspace}

\newcommand{\A}{\mathcal{A}}

\newcommand{\core}{\ensuremath{\textit{CORE}}\xspace}

\newcommand{\lcba}{\emph{\sc LargeCoreBA}\xspace}

\newcommand{\ba}{\emph{\sc KT0-ByzantineAgreement}\xspace}

\newcommand{\IA}{\emph{Implicit Agreement}\xspace}

\newcommand{\old}[1]{}{}

\newboolean{short}
\setboolean{short}{true}

\newcommand{\shortOnly}[1]{\ifthenelse{\boolean{short}}{#1}{}}
\newcommand{\onlyShort}[1]{\ifthenelse{\boolean{short}}{#1}{}}
\newcommand{\longOnly}[1]{\ifthenelse{\boolean{short}}{}{#1}}
\newcommand{\onlyLong}[1]{\ifthenelse{\boolean{short}}{}{#1}}

\newcommand{\I}{\mathcal{I}}

\title{Scalable and Secure Computation Among Strangers: Resource-Competitive Byzantine Protocols}

\author{John Augustine \thanks{Department of Computer Science and Engineering, Indian Institute of Technology at Madras, Chennai, Tamil Nadu, 600036, India. Email: {\tt augustine@iitm.ac.in}. Research supported in part by an Extra-Mural Research Grant (file number EMR/2016/003016) and a MATRICS grant (file number MTR/2018/001198), both funded by the Science and Engineering Research Board, Department of Science and Technology, Government of India and by the VAJRA faculty program of the Government of India.} \hspace{-0.5cm}
\and Valerie King \thanks{Dept. of Computer Science, University of Victoria, Vancouver BC, Canada. Email: {\tt val@uvic.edu}.}
\hspace{-0.5cm}
\and Anisur R. Molla \thanks{Computer and Communication Sciences, Indian Statistical Institute, Kolkata, India. Email: {\tt molla@isical.ac.in}. Research supported by DST Inspire Faculty research grant DST/INSPIRE/04/2015/002801.} \hspace{-0.5cm}
 \and Gopal Pandurangan \thanks{Department of Computer Science, University of Houston, Houston, TX 77204, USA. Email: {\tt gopalpandurangan@gmail.com}. Research supported, in part, by NSF grants CCF-1527867, CCF-1540512,  IIS-1633720,  CCF-BSF-1717075, BSF award 2016419, and by the VAJRA faculty program of the Government of India.} \hspace{-0.5cm}
\and Jared Saia \thanks{Dept. of Computer Science, University of New Mexico, NM, USA. Email: {\tt saia@cs.unm.edu}.
This work is supported by the National Science Foundation grants CNS-1318880 and CCF-1320994.}
}

\begin{document}

\date{}

\maketitle
\thispagestyle{empty}
\begin{abstract}

The last decade has seen substantial progress on designing Byzantine agreement algorithms which are scalable in that they do not require all-to-all communication among nodes. These protocols require each node to play a particular role determined by its ID,  and to send to specific neighbors.  Motivated, in part,  by the rise of permissionless systems such as Bitcoin where arbitrary nodes (whose identities are not known apriori) can join and leave  at will, we extend this research to a more practical model  where each node (initially) does not know the identity of its neighbors.  In particular, a node can send to new destinations only by sending to random (or arbitrary) nodes, or responding (if it chooses) to messages received from those destinations.   We assume a synchronous and fully-connected network, with a full-information, but static Byzantine adversary. A general drawback of existing Byzantine protocols is that the communication cost incurred by the honest nodes may not be proportional  to those incurred by the Byzantine nodes; in fact, they can be significantly higher. Our goal
is to design Byzantine protocols for fundamental problems which are {\em resource competitive}, i.e., the total number of bits sent by all the honest nodes
is not significantly more than those sent by the Byzantine nodes.   

We describe a randomized scalable algorithm to solve Byzantine agreement, leader election, and committee election in this model.  Our algorithm sends an expected $O((T+n)\log n)$ bits and has latency $O(\polylog(n))$, where $n$ is the number of nodes, and $T$ is the minimum of $n^2$ and the number of bits sent by adversarially controlled nodes. The algorithm is resilient to $(\fbad-\epsilon)n$ Byzantine nodes for any fixed $\epsilon > 0$, and succeeds with high probability\footnote{I.e., with probability at least $(1-1/n^c)$ for some constant $c > 0$.}.  Our work can be considered as a first application of resource-competitive analysis to fundamental Byzantine problems.

To complement our algorithm we also show  lower bounds for resource-competitive Byzantine agreement. We prove that, in general, one cannot hope to design Byzantine protocols that have communication
cost that is significantly smaller than the cost of the Byzantine adversary.
 \end{abstract}

\vspace{0.4cm}

\newpage

\setcounter{page}{1}


\section{Introduction}\label{sec:intro}

What happens when you don't know your neighbors?  Anonymity is critical in many modern networks including cryptocurrency~\cite{bitcoin,ethereum}, anonymous communication~\cite{freenet,tor}, and wireless~\cite{kong2004anonymous,li2009privacy,weber2010internet,sicari2015security}.  In anonymous networks, nodes are generally known only by self-generated identifiers\footnote{Such as the public key for a digital signature.}; and communication primitives may be limited to: sending a message to all nodes, sending a message to a random (or arbitrary) node, and responding to a message sent directly.  Unfortunately, all algorithms to coordinate such networks in the presence of malicious faults seem either to require all-to-all communication, or make cryptographic assumptions.

The open nature of permissionless systems such as Bitcoin  allow many nodes to enter the network with little or no admission control.
A major challenge in such systems  is dealing with malicious (also called {\em Byzantine}) nodes, which can try to foil the protocols executed by honest (good) nodes.
Byzantine-resistant protocols are at the heart of  secure and robust networks that can tolerate
malicious nodes, such as P2P networks. Consider the real-world example of Bitcoin --- a decentralized P2P-based digital currency~\cite{bitcoin}. A crucial aspect of Bitcoin is a computational mechanism that allows fault-tolerant agreement on a set of ordered transactions. Agreement in Bitcoin is achieved via a computationally-expensive operation, called {\em mining}.

The problem of achieving agreement under Byzantine faults, \emph{Byzantine agreement}, is a fundamental and long-studied problem in distributed computing~\cite{PSL80,AW,LynchBook}.  In the Byzantine agreement problem,
all good nodes start with an input bit, and we must ensure two conditions: (1) {\em All} good nodes output the same input bit ({\em consensus condition}) and (2) the common bit should be the input bit of some good node ({\em validity condition}).  This must be done despite the presence of a constant fraction of Byzantine nodes that can deviate arbitrarily from the protocol executed by the good nodes.  Byzantine agreement is a ``keystone" problem in  distributed computing, in that
 it provides a critical building block for creating attack-resistant distributed systems. Its importance can be seen from widespread and continued application in many domains: sensor networks \cite{SP04},
grid computing \cite{AK02}, peer-to-peer networks \cite{REGWZK03} and cloud computing \cite{Wright09}. However, despite intensive research, there has still not been a practical solution to the Byzantine agreement problem for large networks. A main reason for this is the {\em large  
message complexity} of currently known protocols, as has been suggested by many systems papers \cite{AF03,ADK06,castro2002practical,MR97,YHET05}. The best known protocols have quadratic message complexity, i.e., $\Theta(n^2)$, where $n$ is the number of nodes in the network.

King and Saia \cite{KS10} described the first Byzantine agreement algorithm in {\em synchronous complete networks} that {\em breaks} the quadratic message barrier {\em under the assumption that nodes a priori know the identities of all their neighbors}. This  assumption  is called
the $KT_1$ model \cite{peleg}, where it is assumed that each node has knowledge of the 
identities of its neighbors\footnote{Just the identities of the neighbors, not any other information such as the internal states of the neighbors is assumed.} a priori. This is in contrast to the $KT_0$ model \cite{peleg}, another standard model where nodes do not know the identity of the neighbors.  In the $KT_1$ model, \cite{KS10}  presented an algorithm where {\em each}  processor sends only $\tilde{O}(\sqrt{n})$ messages, and thus the total message complexity is bounded by $\tilde{O}(n^{1.5})$. This was later improved by Braud-Santoni et al. \cite{braud2013fast} to $O(n \text{ polylog}(n))$ total
message complexity, however, this protocol might require some
node to send $O(n)$ messages. 

The $KT_0$ model seems more applicable to modern, permissionless networks.  While we can convert algorithms for $KT_1$ to $KT_0$ by including an initial step where each node communicates with all its neighbors to obtain their identities, this incurs a $\Theta(n^2)$ message cost.  Hence a fundamental question is:  {\bf Can we design Byzantine protocols that require sub-quadratic messages in the $KT_0$ setting?}

In this paper, we take a step toward addressing the above question.  Our focus is on the fundamental problems of Byzantine agreement, leader election and committee election.  Our main result is an algorithm to solve these problems  while sending a number of bits that is $O((T+n) \log n)$, where $T$ is the minimum of $n^2$ and {\em the number of bits sent by adversarially controlled (Byzantine) nodes}, and $n$ is the network size.  This kind of result where algorithmic cost is measured with respect to adversarial cost, is an example of \emph{resource-competitive analysis}~\cite{Bender:2015:RA:2818936.2818949,gilbert:resource}.  To the best of our knowledge, our result is the first of its kind that introduces resource-competitive analysis to the study of Byzantine agreement and related problems. In particular, our result shows that  Byzantine protocols can be designed
that compete well with the resources (messages) expended by
the Byzantine nodes; if they send less messages then the protocol also sends less. An alternate way
to interpret our result is that Byzantine nodes have to incur significant message complexity  (up to quadratic in $n$)
in order to make the honest nodes to have large message complexity. We note that prior work on Byzantine protocols
all incurred quadratic message complexity (in the $KT_0$ setting) regardless of the behavior of the Byzantine nodes.
Our protocol is efficient, lightweight, and fast (has low
latency) and can be used as a building block for designing
secure and scalable systems.

\medskip
\noindent
\subsection{Model}
\label{sec:model}
We consider a network of $n$ nodes: $t$ are \emph{bad} and controlled by the adversary, and the remainder are \emph{good} and follow our algorithm; we assume $t \leq (\fbad - \efbad)n$ for some constant $\efbad>0$. 

We consider a synchronous, fully-connected network in the $KT_0$ model~\cite{dnabook,peleg}.  In particular, we assume that a node has ports to every other node in the network, but learns the identity of each node reachable through a port only by receiving a communication from that node. Thus a node sends to a new destination only by  selecting a port, or by responding to messages received.  The $n$ nodes are assumed to have distinct ID's which lie in $[1,n^k]$ for $k$ is a (large) constant.\footnote{This means that an ID can be represented using $O(\log n)$ bits, which can be sent in a message. We  assume the CONGEST model, i.e., only $O(\log n)$-sized messages are used in our algorithm.}  Our adversary is full-information in that it knows the states of all nodes at any time, is assumed to be computationally unbounded, and is also \emph{rushing} in the sense that it can read messages sent by good nodes before sending out its own messages. However, the adversary is \emph{static}, so that it must decide which nodes are bad prior to the start of the algorithm. We assume that Byzantine nodes cannot fake their
own identities, however they can forward fake messages on behalf of other nodes. 

\subsection{Our Contributions}
\label{sec:results}

We solve three classic problems in this model.  In Byzantine agreement, all good nodes must output the same bit, which is the input bit of some good node. In leader election, all good nodes must agree on a leader, and this leader must be good with constant probability. In committee election, all nodes must agree on a subset of $O(\log n)$ nodes where the fraction of bad nodes in the subset is within a small $\epsilon$ fraction of the overall fraction of bad nodes.

  Our main result is as follows.

\begin{theorem}\label{t:main}
	There exists a randomized algorithm that solves Byzantine agreement, leader election and committee election in the above model.  This algorithm sends an expected $O((T + n) \log n)$ messages, and has latency $O(\polylog(n))$, where $T$ is the minimum of $n^2$ and the number of bits sent by the bad nodes. It is resilient to $t \leq (\fbad-\efbad)n$ Byzantine faults for any fixed $\efbad>0$, and succeeds with probability $1-1/n^c$ for any constant $c$.  
\end{theorem}

We note that our $O(\polylog(n))$ latency bound holds even in the CONGEST model, where each message is $O(\log n)$ bits. The algorithm \ba described in Section~\ref{s:ba} achieves the result in Theorem~\ref{t:main}, and the proof of this theorem is in Section~\ref{s:main-analysis}.

To complement the above result we also show lower bounds for resource-competitive Byzantine agreement (see  Section \ref{sec:lower-bound}).
 We prove that, in general, one cannot hope to design Byzantine protocols that have communication
cost that is significantly smaller than the cost of the Byzantine adversary, i.e., the no. of messages send by bad nodes. We first show a lower bound for {\em deterministic} BA protocols which is essentially tight with respect to the upper bound of our resource-competitive randomized algorithm (see  Section \ref{sec:det-lb}). We show that if $T = O(n^2)$ is the budget on the message bits of the Byzantine nodes,  then
for any deterministic  protocol, the total number of messages sent by the good nodes is $\Omega(T)$ (see  Theorem \ref{thm:deterministic-lb}).
The deterministic  lower bound holds even in the $KT_1$ model. We then show a somewhat weaker lower bound on the resource competitiveness of randomized BA protocols (see  Section \ref{sec:rand-lb}). The argument for the randomized case is more involved compared to the deterministic case, as the algorithm's (future) random choices are unknown to the Byzantine adversary. We show that if  $T = n^{1+\alpha}$  for some $\alpha \in (0,1]$ is the budget of the Byzantine nodes, then for  any (randomized) BA algorithm  in the  $KT_0$ setting, the total expected number of   messages sent by good nodes, is at least $\Omega(n^{1+\frac{\alpha}{2}})$ (see  Theorem \ref{thm:rand-lb}). Another significance of this  lower bound result
is that it separates  the message complexity of Byzantine agreement
between $KT_0$ and $KT_1$ models in the randomized setting.

\subsection{Techniques and other results}

We focus first on Byzantine agreement, our solutions to leader and committee election use similar techniques.  Our algorithm depends on solutions to two new problems: \IA and \prom.  In the \IA problem, success means that strictly greater than a $t/n$ fraction of good nodes decide on the same (correct) bit and the remaining good nodes do not decide; and failure means that no good nodes decide.  Next, the \prom problem assumes there has first been either success or failure in \IA.  In the case of success, \prom ensures all nodes  decide on the same value and terminate; in the case of failure, no nodes decide.  

\ba runs in epochs.  In each epoch, we (1) run an algorithm for \IA; (2) run an algorithm for \prom; and (3) terminate in the case of success, or increase computational effort in the case of failure.  

The computational effort for \IA is tuned by increasing the number of \emph{active} nodes.  In particular, during a run of \IA, the active nodes first attempt to solve Byzantine agreement among themselves, and then to communicate the output to all other nodes in the network.  Our \IA algorithm ensures that, unless the bad nodes send a number of messages that is $n$ times the number of active nodes, then \IA will succeed.  Next, we solve \prom.  This ensures that if \IA succeeded, then all nodes will decide on the same value and terminate; and if \IA failed, then no nodes decide.  In the latter case, all nodes proceed to the next epoch, where the number of active nodes doubles in expectation.

\medskip
\noindent
{\bf LargeCoreBA.}  There are several technical challenges in the implementation of this main idea.  The first is to enable Byzantine agreement among the active nodes when the bad nodes do not send out too many messages.   We say that a node $x$ has a view of node $y$ if $x$ knows $y$'s ID and the port to $y$.  With a fair amount of technical work, we show that it is possible to modify an algorithm by King et al.~\cite{king2006scalable} to ensure agreement even among nodes whose views only ``mostly" overlap, provided that the range of all IDs is only polynomially large.  We call this modified algorithm  \lcba, and summarize its properties in Lemma~\ref{l:lcba} below; we believe the result may be of independent interest.  The technical approach we used to prove Lemma~\ref{l:lcba} is (1) a counting argument to show there are not too many bad nodes in the views of key participants; and (2) the use of a sampler graph to ensure that these bad participants are well spread over the committees, as used in~\cite{king2006scalable} (see Section~\ref{s:lcba}).   

\begin{lemma} \label{l:lcba} Let  $G$ be a set of good nodes which wish to come to agreement.  For each $x \in G$, let $S_x$ be the set of nodes in the view of $x$.   Let $B$ be the set of bad nodes in $\bigcup_{x \in G} S_x $. Assume $G \subseteq  \bigcap_x S_x$; $|B| \leq (1-\epsilon) |G|/2 $ for some fixed constant $\epsilon>0$; and all nodes have distinct ID's in $[1,n^k]$.  Then there is an algorithm \lcba which computes almost everywhere agreement with high probability among $(1-1/\log n)$ fraction of nodes in $G$ in  time and communication per node which is polylogarithmic in $|G|+|B|$.
In one more round, if each good node broadcasts to all other nodes, and then each node takes the majority, all nodes will come to agreement using $|G|(|G| +|B|)$ total messages, and latency polylogarithmic in $|G| + |B|$.
 \end{lemma}

\medskip
\noindent
{\bf \emph{\IA.}}  Our solution to \IA is given in Steps~\ref{s:init}~to~\ref{s:maj} of our main algorithm in Section~\ref{s:baCode}.   There are two key technical problems that must be addressed.

First, how do we ensure that each active node $x$ maintains a set $S_x$ so that the conditions of Lemma~\ref{l:lcba} are matched? Also, in order to achieve a good competitive ratio, we need the conditions of Lemma~\ref{l:lcba} to hold unless the adversary sends $\Omega(nA)$ messages, where $A$ is the number of active nodes.  If each active node $x$ naively adds to $S_x$ all nodes $y$ that it receives an initial message from, then the adversary can add $A$ Byzantine nodes to each $S_x$ while sending only $A^2$ messages.  Thus, we must enlist the aid of non-active nodes to establish the $S_x$ sets.  Initially, each active node sends its ID to \emph{all} nodes. Call a good node \emph{light} if it has received a number of IDs approximately equal to $A$.  Then the light nodes convey information about their $S_x$ sets to the nodes in $S_x$.  They can not send out \emph{all} the IDs in $S_x$, since that would be too many bits.  Instead, they just send out a single random ID, and a node $y$ adds an ID to $S_y$ if it was received from ``enough" (i.e. $\Theta(nA)$) nodes that claim to be light.  

Unfortunately, an adversary can still cause problems by making the size of the union of the bad nodes in each $S_x$ large, so that $|B|$ is large in Lemma~\ref{l:lcba}, even when the advesary does not send out too many messages.  To solve this problem, we use a ``validation" step, whereby each active node, for each ID in $S_x$, queries $\Theta(\log n)$ random nodes about whether they have the ID in their $S_x$ sets, and filters out the ID unless enough of these queries are answered affirmatively.  Based on information obtained during this process (Step~\ref{s:init}~through Step~\ref{s:query-answering} in Section~\ref{s:baCode}), the active nodes determine if the number of light nodes is sufficient for favorable success in this epoch.

This brings us to the second problem.  How can the active nodes agree on one of two options for this epoch: (1) conditions are favorable for agreement; or (2) conditions are not favorable?  We can make use of \lcba in coming to agreement on an option.  However, this is still challenging given that, under certain conditions, some active nodes may run $\lcba$, while other active nodes may not even have a small enough $S_x$ set to run it.  To address this issue requires careful decisions about whether a node will run \lcba, what its input will be, and whether or not it will trust the output, all based on the node's estimate of the number of light nodes (See Step~\ref{s:proceed}, Section~\ref{s:baCode} for details).  In particular, nodes will sometimes run \lcba, because other nodes are relying on them to do so, even when they plan to ignore the output.  If active nodes decide conditions are favorable via the first call to \lcba (Step~\ref{s:proceed}), they will all run it again (Step~\ref{s:BA2}) to decide on a bit.  Lemma~\ref{l:coreBA} in Section~\ref{s:main-analysis} shows that no matter what the number of light nodes, these two steps ensure all active nodes come to agreement on the same decision.

Finally, in Step~\ref{s:maj}, active nodes send their decision to all other nodes.  Nodes that have small $S_x$ sets take the majority of the messages received in this step, whereas other nodes default to a decision to wait for the next epoch.  We can thereby guarantee the post-condition for \IA: either (1) a strictly greater than $t/n$ fraction of good nodes decide, or (2) no good nodes decide.  We obtain this result even when the adversary floods some good nodes but not others.

\medskip
\noindent
{\bf \emph{\prom.}}
A final technical challenge is to determine whether or not we need to run another epoch.   After solving \IA, either (1) strictly greater than a $t/n$ fraction of the good nodes have decided on the same correct bit; or (2) no good nodes have decided.  We must then ensure that \emph{all} good nodes decide either to terminate or to run another epoch.  To do this, we run an algorithm, \promAlg that solves the \prom problem (see Section~\ref{s:prom}). The solution simply has each node sample a logarithmic number of other nodes, and take a majority vote. It does not increase the overall asymptotic number of messages sent, but some non-active nodes can be forced by the adversary to respond to $O(n)$ requests.  If the outcome of \promAlg is not agreement then all nodes proceed to the next epoch, where the number of active nodes doubles in expectation.  In this way, we can guarantee that that \ba succeeds within $\log(n)$ expected epochs.


\subsection{Related Work}
{

 \paragraph{$KT_0$ and $KT_1$ Communication Model.}

$KT_0$ and $KT_1$ models are two well-studied standard models in distributed computing (e.g., see \cite{disc18,dnabook,peleg}).
It turns out that message complexity of a distributed algorithm
 depends crucially  (as explained below) on the initial knowledge of the nodes.\footnote{It is not hard to see that one can run the algorithms in the $KT_1$ model after one round of communication between neighbors, (which is all-to-all communication in a complete network).}
 For example, consider the situation where there are no Byzantine nodes. It is known that  $\Omega(n)$ expected messages are needed for  explicit leader election\footnote{For explicit leader election,
 where all nodes should know the identity of the elected leader.}  in a complete (fully-connected network) in the $KT_0$ model (see e.g., \cite{AMP18,TCS}), whereas in the $KT_1$ it takes no communication at all (since all nodes know each other's IDs, the minimum one can be selected).  Similarly it has been shown for implicit leader election and implicit agreement\footnote{In implicit leader election, only the leader node should know that it is the leader. In implicit agreement,
 it is enough if a non-empty subset of nodes agree.} that $\Omega(\sqrt{n})$ expected messages is a lower bound in the $KT_0$ model \cite{AMP18}.   It is known   that for various fundamental problems such as broadcast, spanning tree construction, minimum spanning tree construction there is a significant gap in the message complexity between the two models. For example, for all the above problems in an arbitrary graph with $m$ edges, it is known that $\Omega(m)$ is a message lower bound in the $KT_0$ model
 which holds even for randomized (Monte Carlo) algorithms. However, in the $KT_1$ model, this lower bound can be breached:
 all these problems can be solved using randomized algorithms in $\tilde{O}(n)$ messages  \cite{DBLP:journals/jacm/KuttenPP0T15}. For the complete network case, it is known that minimum spanning tree construction needs 
 $\Omega(n^2)$ messages in the $KT_0$ model, while it can be accomplished in $\tilde{O}(n)$ messages in the $KT_1$ model.
   The recent work of Gmyr and Pandurangan \cite{disc18} gives several new algorithms in the $KT_1$ model and shows many other separations between the two models.

\paragraph{Resource-Competitive Analysis.}  This paper introduces {\it resource-competitive analysis}~\cite{Bender:2015:RA:2818936.2818949,gilbert:resource} to the study of Byzantine agreement.  In resource-competitive analysis, the computational cost of the attacker, $T$, is incorporated as a parameter in performance analysis.   That is, the cost of executing an algorithm over a network of $n$ nodes is measured not only as a function of $n$, but also as a function of $T$.

Resource competitive analysis has been applied to designing algorithms for: jamming-resistant wireless communication~\cite{gilbert:near,gilbert:making,king:conflict}; attack-resistance on  multiple access channels~\cite{bender:how}, tolerating adversarial channel noise~\cite{aggarwal2016secure,daniICJournal17,ICALP15}, and efficiently distributing bridges for anonymity networks such as TOR~\cite{zamani2017torbricks}.  See~\cite{Bender:2015:RA:2818936.2818949,gilbert:resource} for detailed surveys.

Related notions comparing the resources of good and bad nodes have been considered earlier, see e.g., the work
of \cite{GMPY06}, that  shows how to achieve fairness in secure computation (either both good and bad parties obtain the result of the computation or nobody does) in the dishonest majority setting with a constant competitive ratio in terms of computational costs/number of steps.

Other related works include  \cite{BGT13} that focus on scalability of secure protocols and show how to achieve sublinear communication.

\paragraph{Byzantine Agreement and Election.}
Byzantine agreement enables participants in a distributed network to reach agreement on a decision, even in the presence of a malicious minority.  Thus, it is a fundamental building block for many applications including: cryptocurrencies~\cite{BonneauMCNKF15,eyal2016bitcoin,GiladHMVZ17,cryptoeprint:2015:521}; trustworthy computing~\cite{cachin:secure,castro1998practical, castro2002practical,1529992,clement-making,kotla2007zyzzyva,SINTRA}; peer-to-peer networks~\cite{adya:farsite,oceanweb}; and databases~\cite{GPD,preguica2008byzantium, zhao2007byzantine}. 

 In 2006, King, Saia, Sanwalani, and Vee ~\cite{king2006scalable} gave a (randomized) algorithm to solve Byzantine agreement, leader election and committee election problems in a model differing from the one in this paper only in the assumption of $KT_1$ communication.  This was the first algorithm to use only $\tilde{O}(1)$ bits of communication per node, and $\tilde{O}(1)$ time to bring almost all processors to agreement.  This result can also be achieved in a particular sparse network \cite{king2006towards}. This initial work produced agreement among all but $o(n)$ nodes.  Further work extended this result to achieve everywhere agreement, while using a number of bits that is $\tilde{O}(n^{3/2})$ (load-balanced)~\cite{king2011load}; and $\tilde{O}(n)$ (not load-balanced)~\cite{braud2013fast}.  All of these algorithms required each node to play a particular role as determined by its unique ID in $[1,n]$,  and to send to specific neighbors. In other words, these algorithms critically rely on the $KT_1$ model. These bounds hold even if the bad nodes send any number of bits.  Establishing Byzantine agreement via the use of committees is a common  approach; for examples, see~\cite{GiladHMVZ17,king2006scalable,Luu:2016}.

\medskip
\paragraph{Paper Organization.} Section~\ref{s:ba} contains \ba. Section~\ref{s:promProblem} formally defines \prom.  Section~\ref{s:main-analysis} analyzes the correctness and cost of \ba and proves Theorem~\ref{t:main}. In Section~\ref{s:lcba} we prove Lemma~\ref{l:lcba}.  In  Section~\ref{s:prom}, we describe an algorithm to solve \prom. In Section \ref{sec:lower-bound}, we prove lower bounds.  Finally, we conclude in Section~\ref{s:conc}.

\section{\ba} \label{s:ba}

Here, we describe the main algorithm for resource-competitive Byzantine agreement. It calls \lcba  and an algorithm \promAlg that solves \prom.  A node $x$ calls \lcba with a set of possible participants $S_x$, which may include nodes which do not themselves participate. 

The algorithm below runs correctly with probability $1-1/n^c$ for any constant $c$, when constant $C$ below is  chosen to be sufficiently large, depending on $c$. We let $\epsilon$ be a small constant such that $0< \epsilon< \efbad^2$.  We set $\maxact=(1+\epsilon)p(n-t)$ and $\minact=(1-\epsilon)p(n-t)$ so that w.h.p. the number of \act nodes lies in this range.

We call a good node \act if it sets its state to active in Step~\ref{s:wake}.  
We call a good node {\it light} if the number of IDs received by it from alleged \act nodes in Step \ref{s:wake} is  less than $\maxact + \epsilon pn$.  We use bounds $Low=n-2t-\epsilon n$ and $High =Low + t$ to describe the number of light and purported light nodes.  For $p > 1/C\log n$, if there are at least $Low-t $ light nodes and each sends a random ID from their list of nodes that reported being active in Step \ref{s:wake}, then w.h.p., at  least $\beta=\frac{(1- \epsilon)(Low-t)}{\maxact + \epsilon pn } $ copies of all their common IDs, in particular, the IDs of all \act nodes,  will be received by every \act node. Finally, an element in an \act node $x$'s set $S_x$ is validated when $x$ queries a random set of $C \log n$ nodes and $\delta C \log n$ nodes respond yes.  $\delta = \frac{(1-\epsilon)(Low-t)}{n}$ is chosen so that w.h.p., every ID in \act will be validated but not many ID's of nodes which are bad.

\subsection{Pseudocode for \ba} \label{s:baCode}
\begin{enumerate}

\item \textbf{Initialize:} 
 Every node $x$ sets $p \leftarrow (C \log n)/n$. Each node $x$ sets $\rout_x \leftarrow 0, \rinput_x \leftarrow 0$, and sets its state to $\lnot \act$ and $\lnot \light$. \label{s:init}
 
 \item  \textbf{Nodes become \act and notify others:} 
 With probability $p$, $x$ sets its state to $\act$ and sends its ID to all nodes.   Every node $x$ sets $S_x$ to the set of IDs  received. A node sets its state to $\light$ if $|S_x| \leq \maxact + \epsilon pn$.  \label{s:wake} 

\item \textbf{\Act nodes learn of other \act nodes:} 
\begin{enumerate}  \label{s:sample-intersection}
\item
Every \light node $x$ randomly selects an ID in $S_x$ and sends it to the nodes in $S_x$. \label{s:sendID}

\item
Every active node $x$ sets $n_x$ to be the number of nodes which send to $x$ in Step \ref{s:sendID}. 
If $n_x \geq  Low-t$ then $x$ resets $S_x$ to be the set of IDs which were received from at least $\beta$ nodes.
For each $ID$ in $S_x$, $x$ sends the query $<ID ?>$ to a random set of $C \log n$ nodes. \label{s:query-sending}

\item
Every light node $x$ answers a query $<ID?>$ if ID is in $S_x$ and the query is sent by a node in $S_x$.  An ID in $S_x$ is considered \emph{validated} if $x$ received at least $\delta C \log n$ responses to the query for ID. Each \act node $x$ that sent queries removes from $S_x$ all IDs which are not validated. \label{s:query-answering}

\end{enumerate}
	
\item \textbf{Can we proceed?} Each \act node $x$ with $n_x \geq Low - t$ runs \lcba with the other  
	nodes in $S_x$.    The input bit to \lcba, $\rinput_x \leftarrow 1$ iff $n_x \geq High $. 
	If $n_x \geq Low $  then $\rout_x \leftarrow$ output of \lcba. 	
	\label{s:proceed}  

\item \textbf{Compute Byzantine Agreement} Each \act node $x$ with $\rout_x = 1$ runs \lcba with nodes in $S_x$, with input bit, $\vall_x$, set to the node's initial input bit.  \\
	Node $x$ then sets $\vall_x$ to the output of this \lcba.  \label{s:BA2}
          
\item \textbf{Take Majority:} Each active node, $x$, sends ($\rout_x$, $\vall_x$) to all nodes.  Then, each node $x$ with $n_x \geq Low - t$ sets $\rout_x$ to the majority $\rout$ bit received from nodes in $S_x$.  If this bit is $1$, then $\vall_x$ is set to the majority $\vall$ bit received from nodes in $S_x$. \label{s:maj} 

\item \textbf{\emph{\prom:}}  Each node $x$ runs \promAlg with the tuple $(\rout_x,\vall_x)$, and resets the tuple based on the outcome. \label{s:prom1}
\begin{enumerate}
\item
 If $\rout_x = 1$, then node $x$ terminates and outputs value $\vall_x$; 
 \item
 Else if $p< 1/(C\log n)$, $p$ doubles and $x$ repeats from Step~\ref{s:wake}.\label{s:centeredSample}
 \item  Else \{$pn\geq  n/(C\log n)$\} every node sends to all its neighbors to determine their IDs and all nodes execute \lcba to compute Byzantine agreement. \label{s:largep}
 \end{enumerate} 
\end{enumerate}

\subsection{\emph{\prom}}\label{s:promProblem}

Here we define a variant of the almost-everywhere to everywhere Byzantine agreement problem, which we call \prom.   In Section~\ref{s:prom}, we describe an algorithm, \promAlg, to solve this problem.

\begin{dfn} \label{d:prom-correctness}
	An algorithm is said to solve the \prom problem if it has the following properties.
	\begin{enumerate}
	\item If (i) there is at least a $t/n + 2\epsilon$ fraction of good nodes with tuple $(\rout,\vall) = (1,v)$, for the same bit $v$; and (ii) all remaining good nodes have $\rout$ value of $0$, then all nodes terminate with tuple $(\rout,\vall) = (1,v)$.
	\item If all good nodes have $\rout = 0$, then all nodes terminate with $\rout=0$. 
	\end{enumerate}
\end{dfn}

\section{Analysis of \ba} \label{s:main-analysis}

\subsection{Correctness}

We call one run of all the steps in the \ba algorithm an epoch.  We assume $t \leq (\fbad - \efbad)n$ for and fixed $\efbad>0$.  We also assume that $p< 1/(C\log n)$ except in Step~\ref{s:largep}.
\begin{lemma} \label{l:highprob}
The following events occur \whp in $n$. 
\begin{enumerate}
\item
The number of \act nodes is between $\minact$  and $\maxact$. \label{ls:num-active}
\item
 If there are at least $Low -t$ light nodes, then all active nodes receive at least $\beta$ copies of the ID of every \act node in Step~\ref{s:sendID}. \label{ls:copiesofactive}
\item
 If there are at least $Low -t$ light nodes, then all \act nodes will consider all IDs of \act nodes validated after Step~\ref{s:query-answering}. \label{ls:validated}
\item 
If an ID is contained in the $S_x$ sets of at most $(1-\epsilon)(\delta-t/n)n$ light nodes in Step~\ref{s:query-sending}, then that ID will not be validated.\label{ls:notvalidated}
 \end{enumerate}
\end{lemma}
\begin{proof}
For each of these items there is a random variable $X$ which is the number of successful independent trials.  In each case, we will show that $E[X] \geq C' \log n$ for some constant $C'$.  Then, Chernoff bounds imply that $Pr(|X-E[X]| \geq \lambda E[X]) \leq n^{-c}$, for any fixed $\lambda <1 $, and any fixed $c$, for $C'$ sufficiently large~\cite{mitzenmacher2017probability}. \\

\noindent
1) Let $X$ be the number of \act nodes. Each node of at least $n-t$ good nodes is \act with probability $p$.  Since $t < n/4$ and $p \geq (C\log n)/n$, then $E[X] \geq (C/2)\log n \geq C' \log n$, where $C$ in Step~\ref{s:init} is chosen sufficiently large.\\

\noindent
2) Fix an active ID.  That ID is sent out from a particular light node with probability at least $1/(\maxact + \epsilon pn)$. Let $X$ be the total number of copies sent of that ID. Then the expected number of copies sent out is at least $(Low-t)/(\maxact + \epsilon pn) \geq C' \log n$, for any fixed C', where constant C in Step~\ref{s:query-sending} is chosen sufficiently large. Hence, by Chernoff bounds, \whp, at least $(1-\epsilon)E[X] =\beta$ copies of the fixed active ID are received by each active node.  Finally, a union bound over at most $n$ possible active IDs that could be sent out establishes the result.\\

\noindent
3) Fix an active node $v$ and an ID $u$, such that $v$ queries about $u$ in Step~\ref{s:query-sending}.  Let $X$ be the number of light nodes queried by $v$ about ID $u$ in this step.  Then $E[X] = ((Low-t)/n) C \log n \geq C'\log n$, for any fixed $C'$ for $C$ in Step~\ref{s:query-sending} chosen sufficiently large.  Thus, by Chernoff bounds, \whp, at least $(1-\epsilon)  E[X]=\delta C \log n$ nodes will answer queries for the ID $u$, and so node $v$ will consider $u$ validated.  Finally, taking union bounds over all choices for $v$ and $u$ shows that, \whp, all active nodes will be considered validated by all active nodes.\\

\noindent
4) Fix an active node $v$ and an ID $u$, such that $v$ queries about $u$ in Step~\ref{s:query-sending}.  Assume there are exactly $(1-\epsilon)(\delta -t/n)n$ light nodes that contain the ID $u$ in their $S_x$ sets. Having fewer such nodes only decreases the probability that node $u$ is validated. Let $X$ be the number of light nodes queried by $v$ that answer the query for ID $u$. Then
\begin{eqnarray*}
	E(X) & = & (1-\epsilon)(\delta -t/n) C \log n\\
	& = & \frac{(1-\epsilon)(Low-t)-t}{n} (1-\epsilon)C \log n\\
	& = & ((1-\epsilon)(1-3(t/n)-\epsilon) -t/n) (1-\epsilon)C \log n\\
	& \geq & C' \log n
\end{eqnarray*}

 where the last step holds for any $C'$ for $C$ chosen sufficiently large, provided that $t/n<1/4$. 
Thus, by Chernoff bounds, by setting $\lambda = (1-\epsilon/2)(1-\epsilon) - 1$, we get that $Pr(X \geq (1-\epsilon/2) (\delta - t/n) C \log n) \leq n^{-c}$, for any fixed $c$, for $C$ sufficiently large.  Let $Y$ be the number of bad nodes that are queried by $v$ about ID $u$.  Again by Chernoff bounds, $Pr(Y \geq (1+\epsilon/2) (t/n) C \log n) \leq n^{-c'}$ for any fixed $c'$ for $C$ chosen sufficiently large.  Putting these two facts together shows that the number of nodes in the sample that may answer $v$'s query about ID $u$ is, \whp, less than $\delta$.

A union bound over all nodes $v$ and IDs $u$ completes the proof.\end{proof}

For a fixed epoch, let $\core$ be the set of active nodes that run \lcba in Step~\ref{s:proceed}.  
We show that the nodes participating in \lcba have the desired properties to successfully complete it when there are at least $Low-t$ light nodes. (See Lemma \ref{l:lcba}.)

From Lemma \ref{l:highprob}, we can observe the following.
\begin{lemma} \label{l:core}
If there are at least $Low-t$ light nodes  then w.h.p., we have the following
\begin{enumerate}
	\item Every \act node is in the \core, and therefore $|\core| \geq \minact$. 
	\item For all $x \in \core$, $\core \subseteq S_x$ after Step~\ref{s:query-answering}. 
	\item Let $B$ be the bad nodes in $\bigcup_{x \in \core} S_x$. At the conclusion of Step \ref{s:sample-intersection}, if  there are at least  $Low-t$ light nodes,  $|B|\leq \epsilon' pn$ for any $\epsilon'>0$, and $\frac{|B|}{|CORE|} \leq 1/2-\epsilon''$ for any $\epsilon''>0$.  \label{l:smallBadViews}
\end{enumerate}
\end{lemma}

\begin{proof}
\noindent
1) Follows from the way $n_x$ is set in Step~\ref{s:query-sending}, and also from Lemma~\ref{l:highprob}(\ref{ls:num-active}).\\

\noindent
2) Follows from Lemma~\ref{l:highprob}(\ref{ls:copiesofactive}) and Lemma~\ref{l:highprob}(\ref{ls:validated}).\\

\noindent
3) We note that no bad ID will be present in an $S_x$ unless it is validated, and by Lemma~\ref{l:highprob}(\ref{ls:notvalidated}) w.h.p., this requires at least $$(1-\epsilon)(\delta-(t/n))n$$ light nodes which contain this ID in their $S_x$ sets. 
A light node contains no more than $\maxact(1+\epsilon)$ IDs, of which at least $\minact$ number are \act.
 
Since $\maxact(1+\epsilon)-\minact < 4\epsilon pn $,
there can be at most $4\epsilon pn$ IDs of bad nodes contained in $S_x$ for each light node $x$.  As there are no more than $n$ light nodes, Lemma \ref{l:highprob} implies that the total number of bad nodes 
which are validated is less than 
\begin{eqnarray*}
	(n) \left(\frac{4\epsilon pn}{(1-\epsilon)(\delta-t/n)n}\right)
	 & \leq & \frac{ 4\epsilon pn^2 }{(1- \epsilon)^2 (n-4(\fbad - \efbad)n-
\epsilon n )}\\
& = & \frac{ 4\epsilon pn^2 }{(1- \epsilon)^2 (n+4 \efbad n-
\epsilon n )}\\
& \leq &\frac{ 12 \sqrt{\epsilon} pn }{(1-\epsilon)^2} \\
& < & \epsilon' pn
\end{eqnarray*}

Where the last line follows for any $\epsilon'$, provided that $\epsilon$ is sufficiently small, and $t< (1/4 -\sqrt{\epsilon})n$; and the second to last line follows since $\epsilon$ was chosen such that $\efbad \geq \sqrt{\epsilon}$.
Thus the total number of bad nodes in $\bigcup_{x \in \core} S_x$ is  less than $\epsilon' pn$.

Since $\core$ contains at least $\minact=(1-\epsilon)p(n-t)$ good nodes, $|\core| \geq (1-\epsilon)p(n-t)$, and hence $\frac{|B|}{|CORE|}< 1/2-\epsilon$.
\end{proof}

\begin{figure}[t]
    \centering
    \includegraphics[width=0.8\columnwidth]{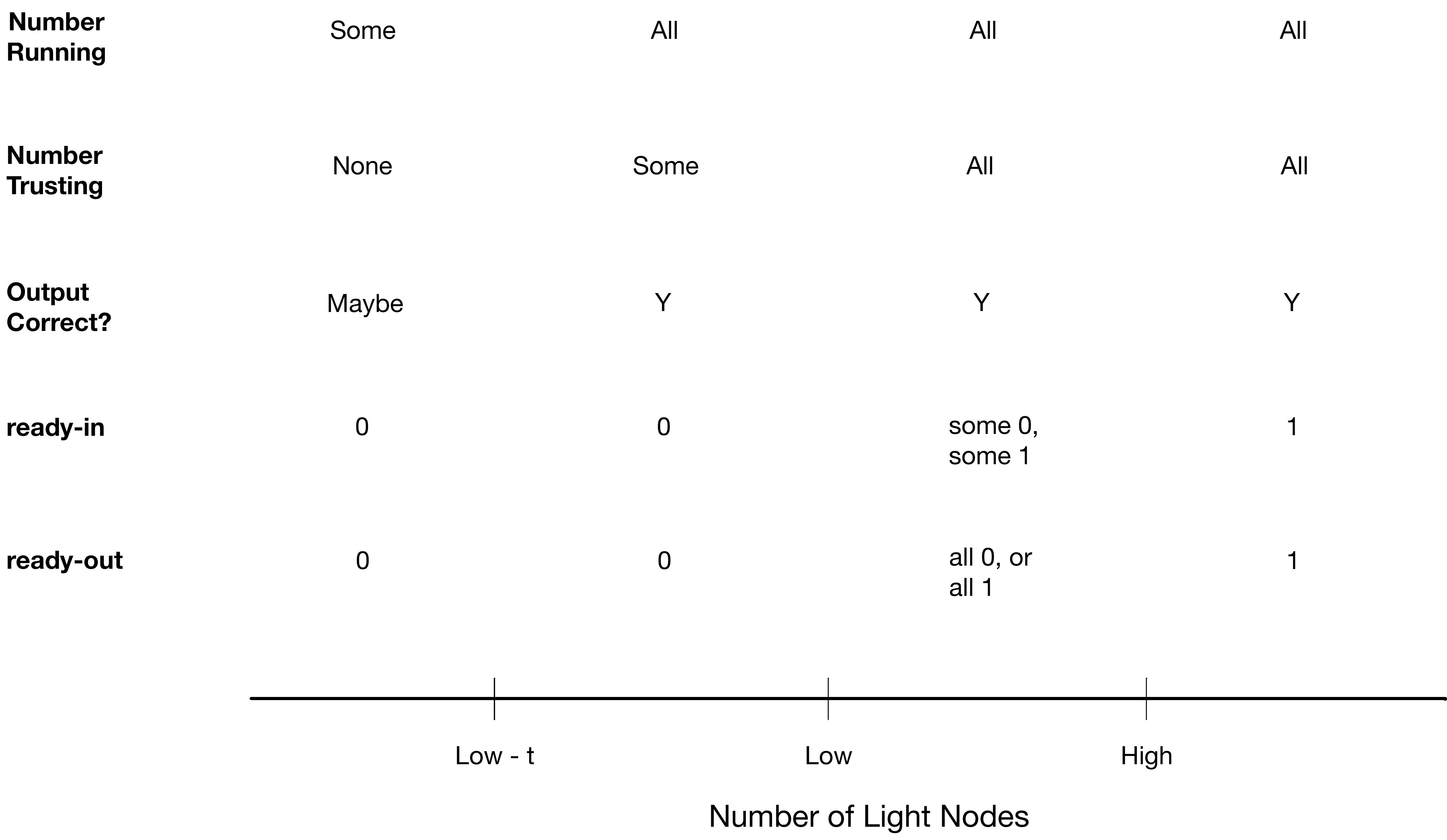}
    \caption{This figure illustrates possible outcomes of \lcba in Step~\ref{s:proceed}, based on what is proven in Lemma~\ref{l:coreBA}.  In each row, the different outcomes depend on the range of the number of light nodes ($L$), as given on the line at the bottom of the figure.  The first row gives the number of \act nodes running \lcba.  The second row gives the number of \act nodes that set their \rout value to the output of \lcba.  The third row says whether or not \lcba actually runs correctly.  The forth row gives the possible $\rinput$ values for the \act nodes.  Finally, the last row gives the possible $\rout$ values for the \act nodes after Step~\ref{s:proceed} is completed.}
    \label{fig:BA-figure}
\end{figure}

Lemma~\ref{l:core} and Lemma~\ref{l:lcba} imply that \lcba can be successfully run when there are at least $Low-t$ light nodes. 
The following lemma follows from this fact and from Lemma \ref{l:highprob}.  Figure~\ref{fig:BA-figure} illustrates part of this lemma.

\begin{lemma} \label{l:coreBA} 
Let $L$ be the number of light nodes in an epoch of \ba.  Then w.h.p.,
\begin{enumerate}
\item
If  $High \leq L $,\\
1)  All \act nodes have $\rinput=1$, they run 
	\lcba and decide on $\rout= 1$ when run in Step~\ref{s:proceed}; and\\
2) All \act nodes $y$ run \lcba in Step \ref{s:BA2} and set their $value$ bit to the input bit $value_x$ of some \act node $x$. 

\item
If  $Low \leq L < High$,\\
1)  All \act nodes successfully run 
	\lcba but they may start with differing values for $\rout$ in Step~\ref{s:proceed}.\\ 
2) If the output is a 1, all \act nodes $y$ set $\rout=1$ and they will successfully run \lcba in Step \ref{s:BA2} and set $value_y$ to the input $value_x$ for some \act node $x$.\\ 
3) If the output is a 0, all \act nodes set $\rout=0$.

\item
If  $Low-t \leq L < Low$, 
all \act nodes  will successfully run \lcba in Step~\ref{s:proceed}, though some nodes will disregard the output. All \act nodes will start with $\rinput =0$ and all \act nodes will have $\rout= 0$. 

\item
If   $L< Low-t$,
some \act nodes  may run a possibly flawed \lcba in Step~\ref{s:proceed}, though all \act nodes will disregard the output. All \act nodes will start with $\rinput=0$ and end with $\rout=0$.
\end{enumerate}
\end{lemma}

\begin{proof}

	If $L \geq Low -t$, then $n_x \geq Low-t$ for all nodes $x$, thus in Step~\ref{s:proceed}, all \act nodes have $\rinput=0$, disregard the output of \lcba, and set $\rout=0$.

	By Lemmas~\ref{l:core} and~\ref{l:lcba}, when $L \geq Low-t$, \lcba will run successfully.  If $L < Low$, then all \act nodes $x$ have $n_x < High = Low +t$, so in Step~\ref{s:proceed}, all \act nodes $x$ have $\rinput_x = 0$.  Thus, by the consistency property of \lcba, all \act nodes $x$ have $\rout_x=0$.
	
	If $L < High$, then all active nodes running \lcba in Step~\ref{s:proceed} may start with different $\rinput$ values, but by the correctness of \lcba, they will all end with the same $\rout$ value.  If the $\rout$ vale is $1$, in Step~\ref{s:BA2},  \lcba will run correctly and they will all set their $value$ bit to the input bit, $value_x$ of some \act node $x$. 

	If $L \geq High$, then any \act node $x$ has $n_x \geq High$, and so has $\rinput_x=1$.  Thus, after Step~\ref{s:proceed}, by the validity of \lcba, all active nodes will have $\rout= 1$.  Thus, they will all run \lcba in Step~\ref{s:BA2} and will all set their value bit to the input value bit of some \act node.
\end{proof}

\begin{lemma} \label{l:agree}
At the end of each epoch, \whp, all nodes either terminate and output the same value or they all go to the next epoch.
\end{lemma}
\begin{proof}
By Lemma \ref{l:coreBA}, if any \act node $x$ sets  $\rout=1$ after Step~\ref{s:BA2}, all \act nodes will set their tuple $(\rout, \vall)$ to the value $(1,v)$, and $v$ will be the input bit of some node in $\core$.  Moreover, there must be at least $Low$ light nodes. Since every light node $y$ has at least $min_a$ IDs of \act nodes in $S_y$, and $|S_y| \leq max_a +\epsilon n$, in Step~\ref{s:maj}, the majority of the messages received from nodes with IDs in $S_y$ will be $(1,v)$ and $y$ will set $\rout=1$ and $\vall_y=v$.  
Since $Low = n-2t+\epsilon \geq t+2\epsilon$, all good nodes will come to agreement on $(1,v)$ in Step \ref{s:prom1}, when the Promise Agreement problem is solved correctly (by Lemma~\ref{l:prom} in Section~\ref{s:prom}).  

On the other hand, if any \act node $x$ sets their value $\rout_x$ to 0, then we must be in Case 2, 3 or 4 of Lemma \ref{l:coreBA}. In these cases, all \act nodes have $\rout=0$, at the end of Step~\ref{s:BA2}. Thus, all light nodes set 
$\rout=0$ since it is the majority value received in Step~\ref{s:maj}, and all nodes which are not light do not change their initial $\rout$ value from 0.  Therefore, all nodes agree on $\rout=0$. With $\rout=0$, all nodes execute Steps \ref{s:centeredSample}  or  \ref{s:largep}, depending on the value of $p$. 
\end{proof}

\subsection{Resource Costs}

\begin{lemma}\label{l:eCost}
In any epoch, \whp, the algorithm sends $O((pn)^2 \log n + pn^2 + n \log n + T_e)$ messages, where $T_e$ is the minimum of $n^2$ and the number of messages sent by bad nodes in that epoch.  Moreover, in any epoch, the algorithm takes time polylogarithmic in $n$. 
\end{lemma}
\begin{proof}
There are $O(pn)$ \act nodes which send to all nodes and each light node sends one message to $O(pn)$ nodes, for a total of $O(pn^2)$ messages.  When $S_x$ is reset, it is reset to be no larger than $n/\beta=O(np)$.
 To validate its $S_x$, each \act node sends $O(\log n)$ messages for each element in $S_x$. There are $O(pn)$ \act nodes, each with $|S_x| =O(n/\beta)=O(pn)$. Hence, issuing  queries requires $O((pn)^2 \log n)$ messages by good nodes.  There are at most $T_e$ queries sent by bad nodes, so responding to queries requires $O(T_e + (pn)^2 \log n)$ messages.
 
 Computing \lcba in Steps~\ref{s:proceed} and~\ref{s:BA2}, requires $O((pn)^2)$ messages by Lemma~\ref{l:lcba}.  Then in Step~\ref{s:maj}, all \act nodes send to all nodes for $O(pn^2)$ messages.  Finally, in Step~\ref{s:prom1}, all nodes send $O(n \log n)$ messages to solve \promAlg, as shown in Lemma~\ref{l:prom}.  Thus, the total number of messages sent in the epoch is $O((pn)^2 \log n + pn^2 + n \log n + T_e)$.
 
 The time to perform all steps in an epoch is dominated by the cost of performing $\lcba$ which is polylogarithmic. 
 \end{proof}

\begin{lemma} \label{l:advcost}
The algorithm terminates in a decision in a given epoch, unless the adversary sends $\Omega(pn^2)$ messages.
\end{lemma}
\begin{proof}
There are at least $High$ light nodes unless the adversary causes bad nodes to send more than $pn \epsilon$ messages to $n\epsilon$ nodes, for a total of $\Omega(pn^2)$ messages.  If there are at least $High$ light nodes in an epoch, then by Lemma~\ref{l:coreBA} the algorithm terminates with a decision. 
\end{proof}
 
 Note that $O((pn)^2 \log n + pn^2)= O(pn^2)$ except when $p> 1/\log n$, in which case our algorithm runs \lcba on all the nodes, by messaging all $n$ of their neighbors, for a total cost of $O(n^2)$. This is the bottleneck in the algorithm which causes it to be $O(\log n)$-competitive instead of $O(1)$-competitive.

Let $T$ be the minimum of $n^2$ and the total number of messages sent by the adversary, and $n$ be the number of nodes in the network.  We can now prove Theorem~\ref{t:main}.

\subsection{Proof of Theorem~\ref{t:main}}

\begin{proof}

By Lemma~\ref{l:advcost}, the algorithm will terminate in an epoch, unless the adversary sends $c pn^2$ messages in that epoch, for some constant $c$.    In epoch $i$, $p = (2^{i-1} \log n)/n$.  If we do not terminate in epoch $i$, then $T \geq c 2^{i-1} n \log n$.  In epoch $i$, by Lemma~\ref{l:eCost}, the total number of messages sent is $O((pn)^2 \log n + pn^2 + n \log n + T_e)$.    

We first consider the case where it's always true that $p \leq 1/\log n$, and note that $O((pn)^2 \log n + pn^2)= O(pn^2)$.  Thus, the message cost in epoch $i$ is $O(n 2^{i} \log n + T_i)$, where $T_i$ is the number of messages sent by the adversary in epoch $i$.  The $T_i$ terms clearly sum to $O(T)$.  If $\ell$ is the last epoch, then $O(\sum_{i=1}^{\ell} 2^i n \log n) = O(2^{\ell} n \log n) = O(T + n \log n)$.  Thus the total number of messages sent in this case is $O(T + n \log n)$.

We next consider the case where $p> 1/\log n$.  In this case, our algorithm runs \lcba on all the nodes, by messaging all $n$ of their neighbors, for a total cost of $O(n^2)$.  The value of $T$ in this case is $\Omega(n^2/\log n)$, so our total message cost is $O(T \log n)$. 

Since epoch $i$ has latency polylogarithmic in $n$ (by Lemma~\ref{l:lcba}), and there are at most $\log n$ epochs, the total latency is $O(\polylog(n))$.  Additionally, we note that when the algorithm terminates, by Lemma~\ref{l:advcost}, all good nodes come to agreement on an input bit of some node in $\core$.

Finally we note that we can also solve the leader election and committee election problems.  To do this, the active nodes use Feige's leader election algorithm to elect a committee in one step, or a leader in $\log^*n$ steps among the $CORE_x$ sets for every active node $x$.  This is done instead of selecting an agreement value as in the KSSV algorithm.
\end{proof}

\section{Additional Algorithms}

\subsection{\lcba} \label{s:lcba}

Here we prove Lemma ~\ref{l:lcba}. We do this by adapting the algorithm from~\cite{king2006scalable}. In that paper, all nodes have a view of all of other nodes and nodes are numbered $[1,n]$. 

The main idea of our adaptation is to show that for any $s$, $\log^{10} n \leq  s \leq n$,  there exists a deterministic assignment of IDs in  $[1,n^k] $ to a set of $s /\ln n$ committees, so that for every subset of size $s$ IDs, a $1-1/\ln^2 n$ fraction of committees are (1) ``sufficiently large"; and (2) contain a nearly representative fraction of both good and bad nodes. 

The algorithm in~\cite{king2006scalable} is built upon a family of bipartite graphs with expansion-like properties.  The existence of such graphs are proved using the probabilistic method (see Section 3 of \cite{king2006scalable}). We need the same properties here, but for a possibly much smaller subset of $s \leq n$ identities, which come from a much larger name space ($[1,n^k]$).  We show that we can start with identities in the range $[1,n^k]$, of which $s$ are active and generate a set of committees which have the required properties with respect to the active nodes, as is needed in each layer of the ``election graph" in 
\cite{king2006scalable}, Corollary 3.2. 

The following lemma achieves the needed result for the  ``static" network in that paper.  We note that, in our algorithm, each node in epoch $i$ knows that the number of active nodes lies in  the range $(1 \pm  \epsilon)2^i \log n$ for a fixed constant $\epsilon$, $0< \epsilon < 1$.  

\begin{lemma}  \label{l:samplerVee}
 Let $k,c, f_g,f_b$ be any positive constants, where $c,k \geq 1$, and $(1-\epsilon) \leq  f_g+f_b \leq (1+\epsilon)$.  Let $i$ be an integer such that $1 \leq i \leq \lg n/ \lg \lg n$, and $s =(1+\epsilon)2^i \log n$.  Further, let $L = [n^k]$, and $R = [s/\ln n]$ and $L_g, L_b \subseteq L$ be any disjoint subsets where $|L_g|  \geq  f_g s $ and $|L_b| \leq f_b s$.  Also for any bipartite graph over nodes $(L,R)$, let $\Gamma(r)$ denote the neighbors of node $r \in R$.  
 
 Then there exists a bipartite graph $(L,R)$, where each node in $R$ has degree $d \geq (C\ln^6 n)(n^k/s)$,  $C= 4kc/(1-\epsilon)$, such that for all but a $1/\ln^2  n$ fraction of nodes $r \in R$, all the following hold, with probability $1-3/n^c$: 
	\begin{enumerate}
\item
 Let $X_s=|\Gamma(r) \cap ( L_g \cup L_b) |$. Then $X_s  \in (1\pm 1/\ln n) (f_g+f_b) C \ln^6 n$.
\item
  Let $X_b=|\Gamma(r) \cap L_b | $. Then $X_b \leq (1 + 1/\ln n) f_b C  \ln^6 n$.
\item
 Let $X_g=|\Gamma(r) \cap  L_g  |$. Then $X_g  \geq  (1 - 1/\ln n )(1-1/\ln n) f_g  C \ln^6 n$.
\end{enumerate}
\end{lemma}

\begin{proof} 
Consider a bipartite multigraph with sides $L$ and $R$ where each node in $R$ has $d=(C\ln^6 n)(n^k/s)$ neighbors chosen uniformly at random with replacement from $L$.
Fix disjoint sets $L_g, L_b \subseteq L$ where $|L_g| = f_g s$ and $|L_b| = f_b s$; and fix  $R' \subseteq R$, $|R'| = s/\ln^3 n \leq |R|/ \ln^2 n$. 

  Let $X$ be the number of edges from nodes in $R'$ to $L_g$. Then 
  $$E(X) = d|R'| \cdot \frac{f_gs}{n^k} = \frac{d f_gs^2}{n^k ln^3 n} = C sf_gln^3 n.$$
  
  By Chernoff  bounds, for any positive $\lambda$, we have $Pr(|X - E(X)| \geq \lambda) \leq 2e^{-\frac{\lambda^2}{2E(X)}}$.  Setting $\lambda = (1/\ln  n) E(X)$, we get
\begin{eqnarray*}
Pr(X \leq (1-1/\ln n)E(X)) & \leq & 2e^{-\frac{(1/\ln^2 n) E(X)}{2}} \\
& \leq & 2e^{-\frac{(1/\ln^2 n )(C s f_g \ln^3 n)}{2}} \\
& \leq & 2e^{-\frac{C sf_g \ln n}{2}}.
\end{eqnarray*}

Let $\xi_{L_g,R'}$ be the event that for all nodes in $R'$, condition (3) above fails. Let $\xi = \cup_{L_g,R'} \xi_{L_g,R'}$, for all $L_g$, $R'$ of appropriate size.  Then by a union bound, 
\begin{eqnarray*}
Pr(\xi) & \leq & {n^k \choose f_g s} {s/\ln n \choose s/\ln^3 n}  2e^{-\frac{C s f_g \ln n}{2}} \\ 
& \leq & e^{f_g s (k \ln n)} \cdot 2^{s /\ln n}   \cdot   2e^{-\frac{C s f_g \ln n}{2}} \\
& \leq &   e^{-\frac{C s f_g \ln n}{2} +  \ln 2 +  s/\ln n + f_g sk \ln n} \\
& \leq & 1/n^{c}
\end{eqnarray*}

The last line in the above holds for $C> 4kc/(1-\epsilon)$. Thus, with high probability,  a random bipartite graph satisfies condition (3) for all sets $R'$ and $L_g$.
A similar analysis shows the same results for conditions (1) and (2).  Putting these three facts together, the probability that a random graph fails any one of these properties is no more than $3/n^{c}$. In particular, such a graph exists. 
\end{proof}

If we regard each node in $R$ as a committee whose members are the set of neighbor nodes contained in $L_g \cup L_B$, and if $|L_g| > (1+\epsilon)(2/3) s$ and $|L_b| \leq (1-\epsilon) s/3 $, the nodes in $L_g$ which are mapped to the same committee can successfully run a linear time deterministic Byzantine agreement algorithm by Dolev et al.~\cite{Dolev82},  even if they do not know the exact number and names of all the bad nodes participating, since the total number of bad nodes in any good node's view mapped to the same committee is less than one third of the good nodes mapped to the committee.   Each good node may fail to send a message to some bad node, or may fail to hear from some bad node, but a Byzantine agreement algorithm is resilient to a bad node which fails to send a message. Thus, the algorithm from~\cite{king2006scalable} will work correctly with high probability.  Moreover, this \lcba algorithm has latency that is polylogarithmic in the number of participants (active ndoes), and requires each active node to send only a polylogarithmic number of bits. 

Note that the algorithm from~\cite{king2006scalable} is almost everywhere, i.e., all but $1/\log n$ fraction of good nodes come to agreement.
If we follow the last step with a step in which all active nodes send to each other, each can take the majority, so that all active good nodes come to the same correct decision, for a total number of messages  $\leq (1 + \epsilon)2^{2i} \log^2 n$ messages in phase $i$.

\subsection{\promAlg} \label{s:prom}
We now present a simple algorithm to solve the \prom problem, defined in Section \ref{s:promProblem}.

\medskip
\noindent
 \promAlg

\begin{enumerate}

\item  Each node $y$ sends a request to  a random set of $c \log n$ nodes. 

\item Each node $x$, upon receiving a request from a node $y$, responds to the request by reporting $(\rout_x, \vall_x)$.

\item If greater than a $t/n +\epsilon$ fraction of nodes sampled by $x$ respond with $\rout=1$, then $x$ sets $\rout \leftarrow 1$ and sets $\vall_x$ to the majority of the \vall bits sent by sampled nodes. Else $\rout_x \leftarrow 0$.
\end{enumerate}

\begin{lemma} \label{l:prom}
\promAlg solves the \prom problem (Definition~\ref{d:prom-correctness}), with $O(1)$ latency, and sending $O(T' + n \log n)$ bits, where $T'$ is the minimum of $n^2$ and the number of messages sent by the adversary during this algorithm.
\end{lemma}

\begin{proof}
Assume there are at least a $t/n+2\epsilon$ fraction of good nodes with $(\rout,\vall) = (1,v)$  for the same bit $v$, and all remaining good nodes have $\rout$ values of $0$.  By Chernoff and union bounds, every good node then has greater than a $t/n+\epsilon$ fraction of good nodes with $\rout$ values of $1$, and less than a $t/n +\epsilon$ fraction of bad nodes in their sample.  Hence, all good nodes will terminate with tuple values of $(\rout,\vall) = (1,v)$.

Assume that all good nodes have $\rout$ values of $0$.  Then by Chernoff and union bounds, each sample has less than a $t/n + \epsilon$ fraction of bad nodes.  Hence, all good nodes will terminate with $\rout$ values of $0$.

The number of bits sent is just the number of queries sent which is $O(T' + n \log n)$.
\end{proof}


\section{Lower Bounds for Resource-Competitive Byzantine Agreement}
\label{sec:lower-bound}

We now study the lower bounds for resource-competitive Byzantine agreement (BA). We first show a tight lower
bound on the resource competitiveness of deterministic BA protocols. Then we show a lower bound on
the resource competitiveness of randomized BA protocols.

\subsection{Deterministic Lower Bound}\label{sec:det-lb}
As per our model (cf. Section \ref{sec:model}) we assume a complete $n$-node network with $\eps n$ Byzantine nodes and $(1-\eps)n$ good nodes (i.e., non-Byzantine) for some small constant $\eps$. We assume the $KT_0$ model.
The Byzantine nodes are controlled by a non-adaptive rushing adversary.
It is assumed that Byzantine nodes cannot fake their own identities. 

In the above setting, the goal is to show a lower bound on the message bits spent by the good nodes in any deterministic algorithm solving Byzantine everywhere agreement. 
The lower bound also holds in the $KT_1$ model, in which a node knows the ID of its neighbors.

The output of the algorithm, i.e., the agreed value depends on the ID, input distribution of the nodes and the information exchanged among the nodes during the execution of the algorithm. More precisely, the output of a node $u$ (with id $ID_u$) is a function $f_u (ID_u,\, b_u, \, X_u) \rightarrow \{0, 1\}$, where the argument $b_u$ is the input bit of $u$ and $X_u$ is the set of received message bits during the execution of the algorithm. Let us call this information $(ID_u,\, b_u, \, X_u)$ as ``transcript" of $u$. The algorithm is deterministic and known to the adversary which controls the Byzantine nodes. Further, the algorithm should work for any input distribution (i.e., the $0-1$ value distribution). Given an input distribution over the nodes, the complete execution of the algorithm is known to the adversary. Based on the execution, the adversary selects Byzantine nodes (in the beginning) in such a way that the algorithm fails to achieve agreement everywhere unless it spends enough messages. In fact, we prove the following result.  

\begin{theorem}\label{thm:deterministic-lb}
Suppose the budget of messages of the Byzantine nodes is $T \leq cn^2$ bits, for some constant $c$. Then any deterministic algorithm, which solves Byzantine everywhere agreement, incurs an expected $\Omega(\min\{T, \, n^2\})$ bits of messages. 

\end{theorem}  
\begin{proof}
Let there be a deterministic algorithm $\A$ that solves the Byzantine agreement everywhere and incurs only $o(T)$ messages. We show a contradiction that the agreement is wrong in the sense that there exists two nodes with two different output value for some input distribution. Consider an arbitrary input distribution $\I$ over $n$ nodes. Since the total messages send by the good nodes is $o(T)$, there must exist a node, say, $u$ that exchanges (sends and receives) less than $\delta T/((1-\eps) n)$ message bits in total for some small constant $\delta < 1$ (the actual value of $\delta$ to be fixed later); otherwise the sum of the messages of all the good nodes would be $\Omega(T)$ (in fact, $u$ is the node which spends minimum number of messages throughout the execution of $\A$). Let $S_u$ be the set of nodes which exchange messages with $u$ throughout the execution of $\A$ on the given input $\I$. Note that, given $\A$ and $\I$, $u$ and $S_u$ are fixed and known to the adversary in the beginning. Further, for different input $\I$, the pair ($u$, $S_u$) might be different. The adversary then selects all the nodes in $S_u$ as Byzantine nodes before the execution starts. Thus the {\em transcript} of $u$ is fully controlled by the Byzantine nodes as $X_u$ is determined by the nodes in $S_u$. The transcript of $u$
is the total history of messages between $u$ and the rest of the nodes. Clearly, the decision of $u$ depends on the choice of $u$'s input  value (0 or 1), $u$'s ID and its transcript (which might also include the IDs of the nodes that it communicated with). Also, in a valid protocol, every node (with every distinct ID and input value) will
have a distinct transcript for deciding 0 or 1, respectively.
Essentially, the adversary can decide a transcript for $u$ (depending on its input value and ID) such that the output value of $u$ would be different than the output value of all other good nodes (assuming all other good nodes execute the algorithm without any influence from the Byzantine nodes). This will give a contradiction to the everywhere agreement.     

We now show that there are enough Byzantine nodes and each Byzantine node has sufficient budget to select all the nodes in $S_u$ as the Byzantine nodes (in the beginning). As explained above, the size of $S_u$ could be at most  $\delta T/((1-\eps) n)$. Clearly $S_u \neq \emptyset$ (otherwise, it won't be a valid protocol). In other words, $1 \leq |S_u| \leq \lceil \delta T/((1-\eps) n) \rceil$. Now consider the following cases.  

\noindent {\bf Case~1:} $|S_u| = 1$. That is $u$ exchanges at most $\delta T/((1-\eps) n)$ messages with a single node, say $w$, throughout the execution of $\A$.  If $w$ is selected to be a Byzantine node, 
it must have a budget of at least $\delta T/((1-\eps) n)$ bits to be sent to $u$. Since the number of Byzantine nodes is $\eps n$ and their total budget is $T$, each Byzantine node can send (up to) $T/\eps n$ messages on average, i.e., the average budget of each Byzantine node can be (up to) $T/\eps n$.  We need, 
$$ \frac{\delta T}{(1-\eps) n} < \frac{T}{\eps n} $$ 
which is satisfied for $\delta < (1-\eps)/\eps$. Note that the remaining Byzantine nodes (i.e., which are not in $S_u$) may behave as the good nodes and response to the good nodes by following the algorithm $\A$. They are chosen from the nodes who spend lesser messages than the other nodes throughout the execution of $\A$. Thus the remaining Byzantine nodes is not spending more messages than the good nodes.      

\noindent {\bf Case~2:} $|S_u| = \lceil \delta T/((1-\eps) n) \rceil$. Then $\eps n$ Byzantine nodes are sufficient to select all the nodes in $S_u$ in the beginning since (by assumption)  $T\leq cn^2$ where the constant $c = (1-\eps) \eps/\delta$ is determined as:
\begin{align*}
& \delta T/((1-\eps) n) \leq \eps n \\
\Rightarrow \,\, & T \leq ((1-\eps) \eps/\delta)n^2
\end{align*}
Further, all the Byzantine nodes in $S_u$ have enough budget to communicate with $u$.  

\noindent {\bf Case~3:} $1 < |S_u| < \lceil \delta T/((1-\eps) n) \rceil$. This  follows from Case~1 and Case~2. Case~1 says that a single Byzantine node has sufficient budget to communicate with $u$. Thus all the nodes in $S_u$ have enough budget to communicate with $u$. Also Case~2 says there are sufficient number of Byzantine nodes to select all the nodes in $S_u$ as Byzantine nodes.     

Thus we claim that, given any input distribution there is a good node $u$ whose output (i.e., agreed value) is completely controlled by the adversary as the adversary controls the transcript of $u$. Now we argue that the algorithm $\A$ fails to achieve agreement everywhere for at least one of the following four input distributions:- $\I_1:$ all the nodes get value $1$, $\I_2:$ all the nodes get value $0$, $\I_3:$ $u$ gets $1$ and rest get $0$, and $\I_4:$ $u$ gets $0$ and rest get $1$. Suppose the output of the algorithm $\A$ is constant, say all the nodes always output $1$, then the agreement is invalid for the input distribution $\I_2$. (Similarly invalid for $\I_1$ if output is always $0$). If the output  is non-constant, i.e., the output depends on the input, IDs and the execution, then we show that there exists two good nodes (one is $u$ and any one from the remaining good nodes) that agree on two different bits for one of the input distributions $\I_3$ and $\I_4$. Based on the input value of $u$  (which is $1$ for $\I_3$, and $0$ for $\I_4$) and its ID, the adversary decides the {\em transcript} of $u$ in such a way that $u$'s output bit will be opposite to the rest of the good nodes' output bit. This contradicts that $\A$ solves the Byzantine agreement everywhere. Thus the number of message bits sent by the good nodes is $\Omega(T)$.     
\end{proof}

The above argument holds even if the IDs are known to the neighbors (i.e., $KT_1$ model). The good nodes, in particular the node $u$, can never determine the Byzantine nodes, since the adversary selects the Byzantine nodes in the beginning and $u$ can have interactions with the Byzantine nodes only.

\subsection{Randomized Lower Bound}\label{sec:rand-lb}   
Let us first consider a complete network with $n$ nodes $V$ in the anonymous setting, i.e., nodes do not have any identifiers. This can be extended to the non-anonymous $KT_0$ setting, where nodes have unique identities. Each node $v$ has $n-1$ ports through which it connects to the $n-1$ other nodes. Thus, if a node sends a message through a port $p \in [n-1]$ to another node $v$, then any message $u$ receives through $p$ is guaranteed to be from $v$. As before,  among the $n$ nodes, a small fraction $\epsilon n$ (assumed to be integral) for a fixed $\epsilon > 0$ are Byzantine and denoted $V^b$; let $V^g = V \setminus V^b$.  Nodes can individually generate uniform and independent random bits as needed, but we do not assume the availability of common coins. 	

Our goal is to   show a lower bound on the message complexity for Byzantine Agreement  in the above setting (everywhere and with success probability 1)  assuming that the adversary has full information and is a rushing adversary. Let us recall the definition of the message complexity. 
\begin{dfn}
For a given BA algorithm $\A$, the message complexity  $M_{\A}$ (or just $M$ when clear from context) is defined as the maximum expected number of the sum of the bits sent by good nodes. The maximum is taken over all possible adversarial strategies (i.e., choice of IDs, port assignments,  input bits, and the behaviour of the Byzantine nodes) and the expectation is over the random bits used by the nodes. 
\end{dfn}

\paragraph{Overview of our approach.} We show that   if bad nodes  can send $\Omega(n^{1 + \alpha})$ messages, then the good nodes must send at least $\Omega(n^{1+\alpha/2})$ messages, for any $\alpha \in (0,1]$. If we assume not (for the sake of contradiction), then, good nodes can  reach agreement while only sending $o(n^{\alpha/2})$ messages on average. Under this situation, when any good node $u$ sends a message to any other good  node $v$, the bad nodes can bombard $v$ with $n^{\alpha/2}$ messages intended for denial of service (DoS). Node $v$ will be unable to distinguish between the legitimate message from $u$ and these DoS messages from bad nodes. As a result, $v$ will have to respond, on average, to $\Omega(n^{\alpha/2})$ messages from bad nodes first. This is more number of messages than what a good node $v$ can afford on average. Thus, several good nodes will not be able to establish two-way contact with any other good node, which we then exploit to show the impossibility via an indistinguishability argument.

\begin{theorem}\label{thm:rand-lb}
Consider any  BA algorithm $\A$ that guarantees that good nodes reach a valid agreement in the anonymous $KT_0$ setting as long as the number of messages sent by Byzantine nodes is at most $B = n^{1+\alpha}$ for some $\alpha \in (0,1]$. Then, the message complexity $M_{\A}$, i.e., the expected number of messages sent by good nodes, is at least $\Omega(n^{1+\frac{\alpha}{2}})$.
\end{theorem}
\begin{proof}
Suppose for the sake of contradiction that there is a BA algorithm $\A$ for which $M_\A \in o(n^{1+\frac{\alpha}{2}})$ despite an adversarial budget of $B$.   A quick upshot is that the average number of messages sent by  good nodes is $o(n^{\frac{\alpha}{2}})$. 

Our argument to show contradiction is structured as follows. We condition our entire argument on the total number of messages sent by good nodes to be bounded by at most $2 M_\A \in o(n^{1+\frac{\alpha}{2}})$, which occurs  with probability at least 1/2 by Markov's Inequality.  
We first observe in Lemma~\ref{lem:no-overload} that  good nodes receive messages from at most $o(n^{\alpha/2})$ different other good nodes. Exploiting this, we then give three adversarial scenarios. In all these scenarios, the Byzantine adversary employs a strategy in which half the Byzantine nodes use a denial-of-service style attack to suppress responses from good nodes. Under this adversarial strategy, we show in Lemma~\ref{lem:suppress} that most good nodes are unable to establish two-way communication with any other good node. The remaining half of the Byzantine nodes exploit this situation to make it difficult for good nodes to distinguish between different scenarios that require different agreement values.

We begin with the following lemma that holds due to an elementary counting argument.
\begin{lemma} \label{lem:no-overload}
Recall the condition that at most $2 M_\A \in o(n^{1+\frac{\alpha}{2}})$ messages are sent by good nodes. Under this condition, with high probability, each good node receives messages from at most $O(M_\A/n + \log n)$ different good nodes. 
\end{lemma}
\begin{proof}
Fix a good node $u$. Let $X_v$ be an indicator random variable taking the value 1 if good node $v$ sent a message to $u$. Notice that $E[X_v] = n_v/n$, where $n_v$ is the number of different nodes to which $v$ sent messages. The number of different good nodes from which $u$ received messages is given by $X = \sum_{v \in V^g} X_v$ and these $X_v$'s are independent owing to the fact that the ports at each node are independently permuted randomly. By linearity of expectation, \[E[X] = \sum_v E[X_v] = (1/n) \sum_v n_v \le 2M_\A/n \in   o(n^{\alpha/2}).\] Thus, we can apply Chernoff bounds to show that $\Pr(X \ge 12 M_\A/n) \le 2^{-2M_\A/n}$.  If $M_\A \ge n\log n$, $\Pr(X \ge 12 M_\A/n) \le 1/n^2$ and applying the union bound over all $u \in V^g$, we get the required result. On the other hand, if $M_\A < n\log n$, then, $E[X_v] <\frac{\log n}{n}$ and $E[X] < \log n$, but we can still apply Chernoff bounds and get \[\Pr(X \ge 12E[X]) \le \Pr(X \ge 12\log n) \le 2^{-2\log n} \le 1/n^2.\] Again, we get the desired bound when we apply the union bound over all $u$. The claim follows when we combine the two Chernoff bounds.
\end{proof}

We now wish to show  a certain adversarial strategy under which, with non-zero probability, at least one node will violate agreement.  The Byzantine adversary's strategy is as follows. 
\begin{description}
\item[Port assignments.] The nodes are interconnected randomly in the sense that each node has port numbers 1 through $n-1$ through which it connects to the $n-1$ other nodes and the adversary connects a random permutation of the nodes to the $n-1$ ports. The permutations used for each of the nodes are (mutually) independent. Of course, the adversary will be aware of the port assignments.

\item[Input bits.] Recall that there are $(1-\epsilon) n$ honest nodes $V^g$. The Byzantine adversary chooses one of the following three different scenarios. 
\begin{description}
\item[Scenario 0:] All good nodes are assigned 0.
\item[Scenario 1:] All good nodes are assigned 1.
\item[Scenario 2:] Exactly half of them $V_0^g$  (chosen uniformly at random)  are assigned an input bit 0 and the rest of them $V_1^g$ are assigned 1. 
\end{description}

\item[Denial of Service (DoS) Attackers.] Out of the $\epsilon n$ Byzantine nodes, the adversary designates $\epsilon n/2$ nodes as {\em DoS attackers} (denoted $D$). 
These DoS attackers behave in a specific manner towards nodes in $V^g$.
We now describe this behavior of the DoS attackers for a fixed node $v \in V^g$.
Whenever a node $u \not\in D$  sends a message $m$ to any good node $v$, the Byzantine adversary can observe this immediately and rushes in with $n^{\frac{\alpha}{2}}$ {\em replicas of} the (same message) $m$ from $n^{\frac{\alpha}{2}}$ DoS attackers denoted $D_{u \to v}$; 
the same $D_{u \to v}$ is used every time $u$ sends a message to $v$. However, for any pair $u\not\in {D \cup \{v\}}$ and $u'\not\in {D \cup \{v\}}$, $u \ne u'$, the adversary seeks to ensure that $D_{u \to v} \cap D_{u' \to v} = \emptyset$.
The adversary can ensure this pairwise disjointness as long as the number of nodes that send messages to $v$ is at most $o(n^{\frac{\alpha}{2}})$. This condition  is in fact satisfied with sufficient probability (cf. Lemma~\ref{lem:no-overload}) as long as the Byzantine adversary ensures that the number of nodes in $V^b \setminus D$ that contact $v$ is at most $o(n^{\frac{\alpha}{2}})$. 
Thus, from $v$'s perspective, to establish two-way contact with some good node $u$, it must intuitively respond to $\Omega(n^\frac{\alpha}{2})$ messages (on expectation), which is asymptotically more than its average budget.  
\item[Byzantine nodes that are not DoS attackers.] The strategy of nodes in $V^b \setminus D$ will be discussed  later. We reiterate that the strategy must ensure that no good node is contacted by more than $o(n^{\frac{\alpha}{2}})$ nodes from $V^b \setminus D$.
\end{description}

A message $m$ from a node $v$ to a  node $u$ is called a {\em response} if it is the first message from $v$ to $u$ and  $u$ had sent a message to $v$ no later than $m$. Under this definition, note that two messages, one from $u$ to $v$ and the other from $v$ to $u$, sent in the same round will be considered responses to each other when no message transpired between them at any earlier point in time. This is however a rare event that we can safely ignore. Furthermore, a response from $v$ to $u$ is called a {\em good response} if both $v \in V^g$ and $u \in V^g$. 

\begin{lemma} \label{lem:suppress}
The total number of good responses sent by nodes in $V^g$ is $o(n)$ with probability at least $1 - o(1)$.
\end{lemma}
\begin{proof}
To aid in this proof, we first prove a claim about a simplified problem that we call the {\em good apples problem} (GAP). In this problem, we have   $\bar{c}$ containers, $\bar{c} \ge c$ but $\bar{c} \in \text{poly}(c)$, and each container has $k$ apples (where $k \in \Omega(c^\epsilon)$ for some $\epsilon > 0$), but exactly one apple is good in each container and the rest are bad. We are allowed to pick a total of $c A(k)$ apples, where $A(k) \in o(k)$ but $A(k) \in \Omega(k^\varepsilon)$ for some $\varepsilon > 0$. Each pick can be from an arbitrarily chosen bin, but the apple that is picked must be chosen uniformly at random (UAR) and without replacement from the apples remaining in the chosen container. The goal is to pick as many good apples as possible. As a preliminary observation, notice that it doesn't help to pick from a container after we have already picked the good apple from it. 
We now show that there is no strategy that guarantees picking $\Omega(c)$ good apples with any reasonable probability.
\begin{claim} \label{clm:gap}
The probability that the number of good apples picked exceeds $\alpha c$ for any fixed $\alpha > 0$ is $o(1)$ regardless of the strategy used.
\end{claim}
\begin{proof}[Proof Sketch]
We argue that the optimal strategy to solve GAP is to greedily pick from a single (arbitrarily chosen) container until we get the good apple in it and then moving to next container (chosen arbitrarily) and repeating. The intuition behind this strategy is that, once we have invested in a container $i$ with at least one pick (while all others are unpicked), a pick from $i$ is more likely to get the good apple than a pick from any other container.

Now under the optimal strategy, we can let $X_i$ be the indicator random variable that is 1 if the number of picks from container $i$ was at least $k/2$ given that the good apple in $i$ was picked. Clearly, $\Pr(X_i = 1) \ge 1/2$. Moreover, given the budget of $cA(k)$ on the number of picks allowed, $\sum_i X_i \le \frac{2cA(k)}{k}$.   Let $G$ denote the total number of good apples that were picked. Notice that $G$ is dominated by $G^*$, the negative binomial distribution where $p=1/2$ and we are required to see $\frac{2cA(k)}{k}$ successes. Clearly, the probability that $G^*$ exceeds $\frac{12cA(k)}{k}$ can be viewed as the probability that the binomial random variable $B(\frac{12cA(k)}{k}, 1/2)$ is no more than $\frac{2cA(k)}{k}$, which is $o(1)$ by Chernoff bound.
\end{proof}

We now return to proving that the number of responses is $o(n)$ with probability at least $1-o(1)$. We can model the problem as a GAP. Each time a node $u \in V \setminus D$ sends a message to $v \in V^g$, recall that nodes in $D_{u \to v}$ send the exact same message. These messages sent by $D_{u \to v} \cup \{u\}$ can be viewed as a container with exactly one good message (i.e., the one sent by $u$) that $v$ needs to respond to. Since the total budget of messages that good nodes can send is $o(n^{1+\alpha/2})$, the total number of such ``containers" can be at most $\bar{c} = o( n^{1+\alpha/2})$ and each container has $k = n^{\alpha/2}+1$ apples (corresponding to the DoS messages and the good message). Since the total message complexity $2 M_\A$ is at most $o(n^{1+ \alpha/2}) = n \times o(n^{\alpha/2})$, we  get the GAP instance by setting $c = n$ and $A(k) = o(n^{\alpha/2})$. From Claim~\ref{clm:gap}, we know that the number of good apples picked is $o(c)$ with probability at least $1-o(1)$, which translates to the claim that the number of good responses is $o(n)$ with probability at least $1-o(1)$.
\end{proof}

The upshot is that a randomly chosen node $v \in V^g$  will not send any good response with probability at least a constant both under Scenario 0 and Scenario 1. Thus, this node (with constant probability) will have only managed to establish one-way communication with other good nodes. Let $I(v) \in V^g$ denote the set of good nodes that sent messages to $v$. By Lemma~\ref{lem:no-overload}, $|I(v)| \in o(n^{\alpha/2})$. Since $\A$ guarantees BA everywhere, $v$ will be able to correctly decide on 0 (resp., 1) for Scenario 0 (resp., Scenario 1).  

Let $\A(v, b)$ denote the strategy used by the nodes in $I(v)$ under Scenario $b$ for $b\in \{0,1\}$. For our lower bounding purpose, we can assume that all decisions by $\A(v,b)$ are made by a single coordinator within $I(v)$. The coordinator executing $\A(v)$  decides the content and timing of the messages sent by nodes in $I(v)$ to $v$. We assume that the coordinator is aware of $I(v)$ from the beginning. We even assume that it is aware of all messages sent by $v$. We are justified in making these assumptions because they only strengthen the algorithm. This strategy $\A(v,b)$ always ensures that $v$ correctly decides on $b$ under Scenario $b$, $b \in \{0, 1\}$. In particular,  it is resilient to all counter strategies by the Byzantine adversary. 

Notice however that $\A(v,b)$ can be executed by $V^b \setminus D$ as well because these strategies only entail the Byzantine adversary playing the role of coordinator, choosing $o(n^{\alpha/2})$ Byzantine nodes and using them to send appropriate messages to $v$ at appropriate times. Of course, the Byzantine nodes are also aware of all messages sent by $v$.

Now, in order to complete the indistinguishability argument, we let the Byzantine adversary  pick two nodes $v_0$ and $v_1$ uniformly at random under Scenario 2 and interact with them in such a way that they decide contradicting values. We condition on $v_0$ and $v_1$ starting with input values 0 and 1, respectively, which is anyway guaranteed to occur with probability 1/4 under Scenario 2. Moreover, the two nodes will not send any good response with constant probability. The Byzantine adversary then uses nodes in $V^b \setminus D$ to execute $\A(v_0, 0)$ and $\A(v_1, 1)$ towards $v_0$ and $v_1$, respectively. Specifically, the Byzantine adversary picks the $o(n^{\alpha/2})$ nodes from $V^b \setminus D$ to execute $\A(v_0, 0)$ (resp., $\A(v_1, 1)$) and -- upon observing messages sent and received by $v_0$ and $v_1$ in each round -- coordinates the chosen nodes to rush in with responses just as $\A(v,0)$ (resp., $A(v,1)$) would towards $v$ under Scenario $0$ (resp., Scenario $1$). Since $v$ decided correctly under both Scenarios 0 and 1, the two nodes must decide on $0$ and $1$, respectively, because their executions in Scenario 2 is indistinguishable from their respective executions in in Scenarios 0 and 1. Thus, with some constant probability, they reach contradicting decision values under Scenario 2.
\end{proof}

\paragraph{Extension to non-anonymous setting}
We finally point out how to extend easily the above result to the non-anonymous setting, where honest nodes
have unique identifiers; however, Byzantine nodes can fake their identifiers.\footnote{We note that this extension does not apply to  setting where the Byzantine nodes are a bit weaker, i.e., they cannot fake their identities.}
This can be done by an easy reduction to the anonymous setting. Suppose in the non-anonymous setting there is
a protocol that violates the lower bound shown in the anonymous setting. 
At the beginning of the protocol in the anonymous setting, honest nodes choose a random ID between $[1,n^3]$; it is easy to see that they are all unique with high probability. They then execute the protocol of the non-anonymous setting.

\section{Conclusion} \label{s:conc}
We have described an efficient randomized resource-competitive algorithm to solve Byzantine agreement, Leader election and Committee election, in the $KT_0$, synchronous communication model, with a static and full-information adversary.
Our algorithm is efficient in the sense that message cost and latency grow slowly with the number of messages sent by the adversary.  In particular, our algorithm uses $O((T + n) \log n)$ bits of communication, and has latency $O(\polylog(n))$, where $T$ is the minimum of $n^2$ and the number of bits sent by the nodes controlled by the adversary. Further, it succeeds with high probability. We also show  lower bounds on resource-competitive Byzantine agreement algorithms. Our lower bounds show that in general, it is not possible to
do significantly better than our algorithm with respect to the number of bits sent by Byzantine nodes.
A key open problem is to   close the gap between upper and
lower bounds for randomized protocols across all budget values.

\bibliographystyle{plain}
\bibliography{security}

\begin{thebibliography}{10}

\bibitem{adya:farsite}
Atul Adya, William~J. Bolosky, Miguel Castro, Gerald Cermak, Ronnie Chaiken,
  John~R. Douceur, Jon Howell, Jacob~R. Lorch, Marvin Theimer, and Roger~P.
  Wattenhofer.
\newblock {\textsc{FARSITE}: Federated, Available, and Reliable Storage for
  Incompletely Trusted Environment}.
\newblock In {\em $5^{th}$ \textsc{USENIX} Symposium on Operating Systems
  Design and Implementation (OSDI)}, pages 1--14, 2002.

\bibitem{AF03}
A.~Agbaria and R.~Friedman.
\newblock Overcoming byzantine failures using checkpointing.
\newblock {\em University of Illinois at Urbana-Champaign Coordinated Science
  Laboratory technical report no. UILU-ENG- 03-2228 (CRHC-03-14)}, 2003.

\bibitem{aggarwal2016secure}
Abhinav Aggarwal, Varsha Dani, Thomas~P. Hayes, and Jared Saia.
\newblock Sending a message with unknown noise.
\newblock In {\em Proceedings of the 19th International Conference on
  Distributed Computing and Networking (ICDCN)}, pages 8:1--8:10, 2018.

\bibitem{ADK06}
Yair Amir, Claudiu Danilov, Jonathan Kirsch, John Lane, Danny Dolev, Cristina
  Nita{-}Rotaru, Josh Olsen, and David~John Zage.
\newblock Scaling byzantine fault-tolerant replication towide area networks.
\newblock In {\em Proceedings of International Conference on Dependable Systems
  and Networks (DSN)}, pages 105--114, 2006.

\bibitem{AK02}
D.~P. Anderson and J.~Kubiatowicz.
\newblock The worldwide computer.
\newblock {\em Scientific American}, 286(3):28--35, 2002.

\bibitem{AW}
Hagit Attiya and Jennifer Welch.
\newblock {\em Distributed Computing: Fundamentals, Simulations and Advanced
  Topics (2nd edition)}.
\newblock John Wiley Interscience, 2004.

\bibitem{AMP18}
John Augustine, Anisur~Rahaman Molla, and Gopal Pandurangan.
\newblock Sublinear message bounds for randomized agreement.
\newblock In {\em Proceedings of the {ACM} Symposium on Principles of
  Distributed Computing (PODC)}, pages 315--324, 2018.

\bibitem{bender:how}
Michael~A. Bender, Jeremy~T. Fineman, Seth Gilbert, and Maxwell Young.
\newblock {How to Scale Exponential Backoff: Constant Throughput, Polylog
  Access Attempts, and Robustness}.
\newblock In {\em Proceedings of the $27^{th}$ Annual ACM-SIAM Symposium on
  Discrete Algorithms (SODA)}, pages 636--654, 2016.

\bibitem{Bender:2015:RA:2818936.2818949}
Michael~A. Bender, Jeremy~T. Fineman, Mahnush Movahedi, Jared Saia, Varsha
  Dani, Seth Gilbert, Seth Pettie, and Maxwell Young.
\newblock {Resource-Competitive Algorithms}.
\newblock {\em SIGACT News}, 46(3):57--71, September 2015.

\bibitem{bitcoin}
Bitcoin.
\newblock Bitcoin website https://bitcoin.org/.

\bibitem{BonneauMCNKF15}
Joseph Bonneau, Andrew Miller, Jeremy Clark, Arvind Narayanan, Joshua~A. Kroll,
  and Edward~W. Felten.
\newblock {SoK: Research Perspectives and Challenges for Bitcoin and
  Cryptocurrencies}.
\newblock In {\em Proceedings of the {IEEE} Symposium on Security and Privacy
  (SP)}, pages 104--121, 2015.

\bibitem{BGT13}
Elette Boyle, Shafi Goldwasser, and Stefano Tessaro.
\newblock Communication locality in secure multi-party computation - how to run
  sublinear algorithms in a distributed setting.
\newblock In {\em Proceedings of the 10th Theory of Cryptography Conference
  (TCC)}, pages 356--376, 2013.

\bibitem{braud2013fast}
Nicolas Braud-Santoni, Rachid Guerraoui, and Florian Huc.
\newblock Fast byzantine agreement.
\newblock In {\em Proceedings of the 2013 ACM symposium on Principles of
  distributed computing}, pages 57--64. ACM, 2013.

\bibitem{cachin:secure}
Christian Cachin and Jonathan~A. Poritz.
\newblock {Secure Intrusion-Tolerant Replication on the Internet}.
\newblock In {\em Proceedings of the International Conference on Dependable
  Systems and Networks (DSN)}, pages 167--176, 2002.

\bibitem{castro1998practical}
Miguel Castro and Barbara Liskov.
\newblock {Practical Byzantine Fault Tolerance}.
\newblock {\em Operating Systems Review}, 33:173--186, 1998.

\bibitem{castro2002practical}
Miguel Castro and Barbara Liskov.
\newblock {Practical Byzantine Fault Tolerance and Proactive Recovery}.
\newblock {\em ACM Transactions on Computer Systems (TOCS)}, 20(4):398--461,
  2002.

\bibitem{1529992}
Allen Clement, Mirco Marchetti, Edmund Wong, Lorenzo Alvisi, and Mike Dahlin.
\newblock {Byzantine Fault Tolerance: The Time is Now}.
\newblock In {\em Proceedings of the Second Workshop on Large-Scale Distributed
  Systems and Middleware (LADIS)}, pages 1--4, 2008.

\bibitem{clement-making}
Allen Clement, Edmund Wong, Lorenzo Alvisi, Mike Dahlin, and Mirco Marchetti.
\newblock {Making Byzantine Fault Tolerant Systems Tolerate Byzantine Faults}.
\newblock In {\em Proceedings of the Sixth USENIX Symposium on Networked
  Systems Design and Implementation (NSDI)}, pages 153--168, 2009.

\bibitem{daniICJournal17}
Varsha Dani, Tom Hayes, Mahnush Movahedi, Jared Saia, and Maxwell Young.
\newblock {Interactive Communication with Unknown Noise Rate}.
\newblock {\em Information and computation}, 261(Part):464--486, 2018.

\bibitem{ICALP15}
Varsha Dani, Mahnush Movahedi, Jared Saia, and Maxwell Young.
\newblock {Interactive Communication with Unknown Noise Rate}.
\newblock In {\em Proceedings of the Colloquium on Automata, Languages, and
  Programming (ICALP)}, pages 575--587, 2015.

\bibitem{Dolev82}
Danny Dolev, Michael~J. Fischer, Robert~J. Fowler, Nancy~A. Lynch, and
  H.~Raymond Strong.
\newblock An efficient algorithm for byzantine agreement without
  authentication.
\newblock {\em Information and Control}, 52(3):257--274, 1982.

\bibitem{ethereum}
Ethereum.
\newblock Ethereum website https://ethereum.org/.

\bibitem{eyal2016bitcoin}
Ittay Eyal, Adem~Efe Gencer, Emin~G{\"u}n Sirer, and Robbert Van~Renesse.
\newblock {Bitcoin-NG: A Scalable Blockchain Protocol}.
\newblock In {\em Proceedings of the $13^{th}$ USENIX Symposium on Networked
  Systems Design and Implementation (NSDI)}, pages 45--59, 2016.

\bibitem{freenet}
Freenet.
\newblock Freenet website
  https://freenetproject.org/author/freenet-project-inc.html.

\bibitem{GMPY06}
Juan~A. Garay, Philip~D. MacKenzie, Manoj Prabhakaran, and Ke~Yang.
\newblock Resource fairness and composability of cryptographic protocols.
\newblock In {\em Proceedings of the 3rd Theory of Cryptography Conference
  (TCC)}, pages 404--428, 2006.

\bibitem{GiladHMVZ17}
Yossi Gilad, Rotem Hemo, Silvio Micali, Georgios Vlachos, and Nickolai
  Zeldovich.
\newblock {Algorand: Scaling Byzantine Agreements for Cryptocurrencies}.
\newblock In {\em Proceedings of the 26th Symposium on Operating Systems
  Principles (SOSP)}, pages 51--68, 2017.

\bibitem{gilbert:near}
Seth Gilbert, Valerie King, Seth Pettie, Ely Porat, Jared Saia, and Maxwell
  Young.
\newblock {(Near) Optimal Resource-Competitive Broadcast with Jamming}.
\newblock In {\em Proceedings of the $26^{th}$ ACM Symposium on Parallelism in
  Algorithms and Architectures (SPAA)}, pages 257--266, 2014.

\bibitem{gilbert:resource}
Seth Gilbert, Valerie King, Jared Saia, and Maxwell Young.
\newblock {Resource-Competitive Analysis: A New Perspective on Attack-Resistant
  Distributed Computing}.
\newblock In {\em Proceedings of the $8^{th}$ ACM International Workshop on
  Foundations of Mobile Computing}, page~1, 2012.

\bibitem{gilbert:making}
Seth Gilbert and Maxwell Young.
\newblock {Making Evildoers Pay: Resource-Competitive Broadcast in Sensor
  Networks}.
\newblock In {\em Proceedings of the $31^{th}$ Symposium on Principles of
  Distributed Computing (PODC)}, pages 145--154, 2012.

\bibitem{disc18}
Robert Gmyr and Gopal Pandurangan.
\newblock Time-message trade-offs in distributed algorithms.
\newblock In {\em Proceedings of the 32nd International Symposium on
  Distributed Computing (DISC)}, pages 32:1--32:18, 2018.

\bibitem{cryptoeprint:2015:521}
Sergey Gorbunov and Silvio Micali.
\newblock {Democoin: A Publicly Verifiable and Jointly Serviced
  Cryptocurrency}.
\newblock Cryptology ePrint Archive, Report 2015/521, 2015.
\newblock http://eprint.iacr.org/2015/521.

\bibitem{king2011load}
Valerie King, Steven Lonargan, Jared Saia, and Amitabh Trehan.
\newblock Load balanced scalable byzantine agreement through quorum building,
  with full information.
\newblock In {\em International Conference on Distributed Computing and
  Networking}, pages 203--214. Springer, 2011.

\bibitem{KS10}
Valerie King and Jared Saia.
\newblock Breaking the \emph{O}(\emph{n}\({}^{\mbox{2}}\)) bit barrier:
  Scalable byzantine agreement with an adaptive adversary.
\newblock {\em J. {ACM}}, 58(4):18:1--18:24, 2011.

\bibitem{king2006scalable}
Valerie King, Jared Saia, Vishal Sanwalani, and Erik Vee.
\newblock Scalable leader election.
\newblock In {\em {Proceedings of the Seventeenth annual ACM-SIAM Symposium on
  Discrete Algorithm}}, pages 990--999. Society for Industrial and Applied
  Mathematics, 2006.

\bibitem{king2006towards}
Valerie King, Jared Saia, Vishal Sanwalani, and Erik Vee.
\newblock Towards secure and scalable computation in peer-to-peer networks.
\newblock In {\em Foundations of Computer Science, 2006. FOCS'06. 47th Annual
  IEEE Symposium on}, pages 87--98. IEEE, 2006.

\bibitem{king:conflict}
Valerie King, Jared Saia, and Maxwell Young.
\newblock {Conflict on a Communication Channel}.
\newblock In {\em Proceedings of the $30^{th}$ Symposium on Principles of
  Distributed Computing (PODC)}, pages 277--286, 2011.

\bibitem{kong2004anonymous}
Jiejun Kong.
\newblock {\em Anonymous and untraceable communications in mobile wireless
  networks}.
\newblock Citeseer, 2004.

\bibitem{kotla2007zyzzyva}
Ramakrishna Kotla, Lorenzo Alvisi, Mike Dahlin, Allen Clement, and Edmund Wong.
\newblock {Zyzzyva: Speculative Byzantine Fault Tolerance}.
\newblock In {\em Proceedings of $21^{st}$ ACM SIGOPS Symposium on Operating
  Systems Principles}, pages 45--58, 2007.

\bibitem{DBLP:journals/jacm/KuttenPP0T15}
Shay Kutten, Gopal Pandurangan, David Peleg, Peter Robinson, and Amitabh
  Trehan.
\newblock On the complexity of universal leader election.
\newblock {\em J. {ACM}}, 62(1):7:1--7:27, 2015.

\bibitem{TCS}
Shay Kutten, Gopal Pandurangan, David Peleg, Peter Robinson, and Amitabh
  Trehan.
\newblock Sublinear bounds for randomized leader election.
\newblock {\em Theor. Comput. Sci.}, 561:134--143, 2015.

\bibitem{li2009privacy}
Na~Li, Nan Zhang, Sajal~K Das, and Bhavani Thuraisingham.
\newblock Privacy preservation in wireless sensor networks: A state-of-the-art
  survey.
\newblock {\em Ad Hoc Networks}, 7(8):1501--1514, 2009.

\bibitem{Luu:2016}
Loi Luu, Viswesh Narayanan, Chaodong Zheng, Kunal Baweja, Seth Gilbert, and
  Prateek Saxena.
\newblock {A Secure Sharding Protocol For Open Blockchains}.
\newblock In {\em Proceedings of the 2016 ACM SIGSAC Conference on Computer and
  Communications Security (CCS)}, pages 17--30, 2016.

\bibitem{LynchBook}
Nancy Lynch.
\newblock {\em Distributed Algorithms}.
\newblock Morgan Kaufmann, 1996.

\bibitem{MR97}
Dahlia Malkhi and Michael~K. Reiter.
\newblock Unreliable intrusion detection in distributed computations.
\newblock In {\em Proceedings of the 10th Computer Security Foundations
  Workshop (CSFW)}, pages 116--125, 1997.

\bibitem{mitzenmacher2017probability}
Michael Mitzenmacher and Eli Upfal.
\newblock {\em Probability and computing: randomization and probabilistic
  techniques in algorithms and data analysis}.
\newblock Cambridge university press, 2017.

\bibitem{GPD}
Hector~Garcia Molina, Frank Pittelli, and Susan Davidson.
\newblock {Applications of Byzantine Agreement in Database Systems}.
\newblock {\em ACM Transactions on Database Systems (TODS)}, 11:27--47, 1986.

\bibitem{oceanweb}
Oceanstore.
\newblock {The Oceanstore Project}.
\newblock http://oceanstore.cs.berkeley.edu.

\bibitem{dnabook}
Gopal Pandurangan.
\newblock {\em Distributed Network Algorithms}.
\newblock https://sites.google.com/site/gopalpandurangan/dna, 2018.

\bibitem{PSL80}
Marshall~C. Pease, Robert~E. Shostak, and Leslie Lamport.
\newblock Reaching agreement in the presence of faults.
\newblock {\em J. {ACM}}, 27(2):228--234, 1980.

\bibitem{peleg}
D.~Peleg.
\newblock {\em Distributed Computing: A Locality Sensitive Approach}.
\newblock SIAM, 2000.

\bibitem{preguica2008byzantium}
Nuno Pregui{\c{c}}a, Rodrigo Rodrigues, Christ\'{o}v\={a}o Honorato, and
  Jo\={a}o Louren{\c{c}}o.
\newblock {Byzantium: Byzantine-Fault-Tolerant Database Replication Providing
  Snapshot Isolation}.
\newblock In {\em Proceedings of the Fourth Conference on Hot Topics in System
  Dependability}, page~9. USENIX Association, 2008.

\bibitem{REGWZK03}
Sean~C. Rhea, Patrick~R. Eaton, Dennis Geels, Hakim Weatherspoon, Ben~Y. Zhao,
  and John Kubiatowicz.
\newblock Pond: The oceanstore prototype.
\newblock In {\em Proceedings of the {FAST} '03 Conference on File and Storage
  Technologies}, pages 1--14, 2003.

\bibitem{SP04}
Elaine Shi and Adrian Perrig.
\newblock Designing secure sensor networks.
\newblock {\em {IEEE} Wireless Commun.}, 11(6):38--43, 2004.

\bibitem{sicari2015security}
Sabrina Sicari, Alessandra Rizzardi, Luigi~Alfredo Grieco, and Alberto
  Coen-Porisini.
\newblock Security, privacy and trust in internet of things: The road ahead.
\newblock {\em Computer networks}, 76:146--164, 2015.

\bibitem{SINTRA}
SINTRA.
\newblock {SINTRA - Distributed Trust on the Internet}.
\newblock http://www.zurich.ibm.com/security/dti/.

\bibitem{tor}
Tor.
\newblock Tor website https://www.torproject.org/.

\bibitem{weber2010internet}
Rolf~H Weber.
\newblock Internet of things--new security and privacy challenges.
\newblock {\em Computer law \& security review}, 26(1):23--30, 2010.

\bibitem{Wright09}
Alex Wright.
\newblock Contemporary approaches to fault tolerance.
\newblock {\em Commun. {ACM}}, 52(7):13--15, 2009.

\bibitem{YHET05}
Hiroyuki Yoshino, Naohiro Hayashibara, Tomoya Enokido, and Makoto Takizawa.
\newblock Byzantine agreement protocol using hierarchical groups.
\newblock In {\em Proceedings of the 11th International Conference on Parallel
  and Distributed Systems (ICPADS)}, pages 64--70, 2005.

\bibitem{zamani2017torbricks}
Mahdi Zamani, Jared Saia, and Jedidiah Crandall.
\newblock {TorBricks: Blocking-Resistant Tor Bridge Distribution}.
\newblock In {\em International Symposium on Stabilization, Safety, and
  Security of Distributed Systems (SSS)}, pages 426--440. Springer, 2017.

\bibitem{zhao2007byzantine}
Wenbing Zhao.
\newblock {A Byzantine Fault Tolerant Distributed Commit Protocol}.
\newblock In {\em Proceedings of the $3^{rd}$ IEEE International Symposium on
  Dependable, Autonomic and Secure Computing}, pages 37--46, 2007.

\end{thebibliography}

\end{document}